\documentclass[journal]{IEEEtran}

\usepackage{xcolor,soul,framed} %,caption

\colorlet{shadecolor}{yellow}
\usepackage[pdftex]{graphicx}
\graphicspath{{../pdf/}{../jpeg/}}
\DeclareGraphicsExtensions{.pdf,.jpeg,.png}

\usepackage{amsmath}
\usepackage{verbatim}
\usepackage{array}
\usepackage{mdwmath}
\usepackage{mdwtab}
\usepackage{eqparbox}
\usepackage{url}
\usepackage{nccmath}
\usepackage{multirow}
\usepackage{diagbox}
\usepackage{lipsum}
\usepackage{amsthm}
\usepackage{multicol}
\usepackage{pifont}
\newtheoremstyle{mystyle}%                % Name
  {}%                                     % Space above
  {}%                                     % Space below
  {\upshape}%                             % Body font
  {}%                                     % Indent amount
  {\bfseries}%                            % Theorem head font
  {:}%                                    % Punctuation after theorem head
  { }%                                    % Space after theorem head, ' ', or \newline
  {}%                                     % Theorem head spec (can be left empty, meaning `normal')

\theoremstyle{mystyle}

\newtheorem{theorem}{Theorem}

\newtheorem{property}{Property}

\newtheorem{myDef}{Definition}

\usepackage{algorithm}
\usepackage{algorithmic}

\usepackage{fancyhdr}
\fancypagestyle{firstpage}{%
  \lhead{\textbf{This work has been submitted to the IEEE for possible publication.  Copyright may be transferred without notice, after which this version may no longer be accessible.}}
}

\begin{document}

\bstctlcite{IEEEexample:BSTcontrol}
    \title{Joint User Association and Resource Allocation for Tailored QoS Provisioning in 6G HetNets}

\author{Yongqin~Fu,~\IEEEmembership{Graduate Student Member,~IEEE,} and 
      Xianbin~Wang,~\IEEEmembership{Fellow,~IEEE}\\

\thanks{Y. Fu and X. Wang are with the Department of Electrical and Computer Engineering, Western University, London, ON N6A 5B9, Canada (e-mail: yfu335@uwo.ca; xianbin.wang@uwo.ca).}
\thanks{Corresponding author: Xianbin Wang.}
}

% ====================================================================
\maketitle

\thispagestyle{firstpage}

% === ABSTRACT ====================================================================
% =================================================================================
\begin{abstract}
The proliferation of wireless-enabled applications with divergent quality of service (QoS) requirements necessitates tailored QoS provisioning.
With the growing complexity of wireless infrastructures, application-specific QoS perceived by a user equipment (UE) is jointly determined by its association with the supporting base station in heterogeneous networks (HetNets) and the amount of resource allocated to it. However, conventional application-agnostic objective-based user association and resource allocation often ignore the differences among applications' specific requirements for resources, inevitably preventing tailored QoS provisioning. Hence, in this paper, the problem of joint user association and resource allocation with application-specific objectives is investigated for achieving tailored QoS provisioning in 6G HetNets. This problem is intrinsically difficult to solve directly due to the extremely large solution space and the combination of discrete and continuous variables. Therefore, we decompose the original problem into two subproblems, i.e. user association and resource allocation, and propose an interactive optimization algorithm (IOA) to solve them iteratively in an interactive way until convergence is achieved.
Specifically, matching theory is utilized to solve resource allocation and user association is solved heuristically. Extensive experimental results confirm that IOA algorithm outperforms several baseline algorithms in terms of both average utility and UE satisfaction ratio.
\end{abstract}

% === KEYWORDS ====================================================================
% =================================================================================

\section{Introduction}
The past ten years have witnessed the proliferation of wireless-enabled applications with divergent QoS requirements. For example, autonomous driving \cite{10413956} and remote surgery \cite{8933555} applications call for ultra-low latency and high reliability. However, ultra-high definition (UHD) video streaming \cite{9068423} and virtual reality (VR) \cite{8713498} applications have stringent requirements in terms of data transmission rate. 

With the evolution of wireless communication technologies, in order to support UEs with different radio access technologies (RATs), different types of base stations (BSs) have been deployed in real-world scenarios, including macro, micro, pico and femto base stations \cite{7994919}. These different types of BSs vary from each other in terms of coverage area, power constraint and so on \cite{6692187}. These different types of BSs form a heterogeneous network (HetNet), which enhances coverage, capacity and energy efficiency \cite{6133859}.

Hence, if a UE is located within an area covered by multiple base stations, it should be decided which base station provide service to this specific UE. This process is commonly known as user association \cite{10234545}. In HetNets, the amounts of available resources of each base station are limited. Hence, when traffic load is high, user association plays a vital role for achieving load balancing among base stations \cite{10419174}, which is necessary for satisfying UEs' QoS requirements. For example, if user association result is extremely unbalanced, some UEs' perceived QoS values might be very bad.

Once user association decision is made, each base station need to decide the amounts of resources to be allocated to each UE associated with it, commonly known as resource allocation \cite{9973061}. Resource allocation directly influences UEs' perceived QoS values. Hence, since that user association and resource allocation jointly influence UEs' perceived QoS values, it is necessary to jointly consider the two tasks in order to better satisfy UEs' QoS requirements.

The diverse QoS requirements of applications bring challenges to user association and resource allocation in HetNets, since that they inevitably complicate the problems of user association and resource allocation and make it even more difficult to satisfy UEs' QoS requirements. 

\subsection{Motivations}
The existing user association and resource allocation techniques were usually designed with the aim of optimizing one or several application-agnostic objectives, such as sum rate, energy efficiency and so on. This might be problematic in the 6G era due to the following two reasons.

First, application-agnostic objective-based user association and resource allocation might cause waste of resources. Let's consider the scenario where there exist two groups of UEs. The first group of UEs have low QoS requirements in terms of data transmission rate but have good channel conditions. The second group of UEs have high QoS requirements in terms of data transmission rate but have poor channel conditions. Then conventional user association and resource allocation techniques with the objective of maximizing sum rates tend to allocate a lot of resources to the first group, far beyond their needs, while allocating the amount of resources to the second group which could not or only meet their minimum QoS requirements. Evidently, this is a waste of resources. 

Second, application-agnostic objective-based user association and resource allocation techniques might bring unfairness among different types of applications. In conventional user association and resource allocation techniques with the objective of maximizing energy efficiency, UEs with autonomous driving application might only receive the minimum amount of physical resource blocks (PRBs) which could just satisfy their minimum QoS requirements, because the higher requirement of reliability inevitably results in lower energy efficiency. However, UEs with UHD video streaming application might be allocated much more PRBs, since that their relative lower requirement of reliability results in relative higher energy efficiency. Absolutely, the fairness between the two applications can hardly be regarded as having been guaranteed.

One common cause of the two drawbacks of conventional application-agnostic objective-based user association and resource allocation techniques is the neglect of differences among demand degrees on resources of different applications, which apparently could not be reflected by one or several application-agnostic objectives.

Therefore, to avoid waste of resources and ensure fairness among different types of applications, tailored QoS provisioning is essential. In tailored QoS provisioning, the differences among demand degrees on resources of different applications are highlighted. Specifically, the optimization objective of user association and resource allocation is shifted from application-agnostic objectives to application-specific objectives. By optimizing the application-specific objectives, the fairness among different types of applications could be better guaranteed. 

\subsection{Related Literature}
Due to the significance of joint user association and resource allocation in downlink HetNets, a lot of research efforts have been devoted to this problem within recent years. It needs to be pointed out that resource allocation considered in this paper involves both spectrum and power resources. The related literature published in recent years are compared in Table \ref{RelatedLiterature}. 

Some researchers aim at maximizing throughput or sum rate or weighted sum rate. 
In \cite{wang2017joint}, the problem of joint user association and resource allocation in a multi-cell multi-association OFDMA HetNet is decomposed into two subproblems and solved alternatively by Hungarian algorithm and the difference of two convex functions approximation (DCA) method, respectively.  The authors in \cite{chen2015joint} propose a distributed belief propagation algorithm for joint user association and resource allocation in downlink HetNets, considering intercell interference coordination. In \cite{10211332}, the user association and resource allocation problem in HetNets is abstracted into a multiple 0/1 knapsack problem without considering interference among different base stations and a performance-improved reduced search space simulated annealing algorithm is proposed. The authors in \cite{cheng2020joint} take user mobility prediction into consideration for joint user association and resource allocation in HetNets and propose a coupling solution and a decoupling solution utilizing multi-agent Q-learning method. Moreover, deep Q-network is applied for accelerating convergence. In \cite{10184114}, the authors investigate the problem of joint user association and beam selection in a HetNet consisted of one macro base station and some small base stations with mmWave frequency, aiming at maximizing sum rate. Moreover, a graph neural network (GNN)-enabled solution with a primal-dual learning framework is proposed.

Some researchers aim at maximizing the sum of logarithmic rates or weighted logarithmic rates. 
The authors in \cite{han2016backhaul} investigate the problem of backhaul-aware joint user association and resource allocation in a HetNet where energy is constrained, which is solved by a distributed algorithm based on Lagrangian dual decomposition. In \cite{vaezpour2019robust}, the problem of joint user association and resource allocation in H-CRAN with dual connectivity is studied. And a block coordinate descent-based algorithm is proposed to solve this problem, utilizing successive convex approximation.   In \cite{liu2019joint}, the problem of joint user association and resource allocation in multi-band millimeter-wave HetNets is investigated, considering both single-band access scheme and multi-band access scheme. 

Some researchers aim at maximizing spectrum or/and energy efficiency. In \cite{somesula2022artificial}, a modified artificial bee colony algorithm is proposed to jointly solve user association and resource allocation in a heterogeneous cloud radio access network (H-CRAN).  The authors in \cite{liu2020joint} propose a Lyapunov optimization framework to jointly solve user association and resource allocation aiming at maximizing spectrum efficiency and energy efficiency while considering stability. In \cite{MUGHEES2023102206}, a multi-agent parameterized deep reinforcement learning approach is proposed for joint user association and power allocation in 5G HetNets. The authors in \cite{10246069} propose a parameterized double deep Q-network for joint user association and power allocation in a two-tier HetNet. 

Some researchers take power consumption into consideration for the optimization objective. In \cite{li2016joint}, the problem of base station operation, user association and resource allocation are jointly optimized aiming at minimizing the average energy consumption of the network, considering the dynamicity of spatial and temporal traffic. In \cite{yin2019energy}, the problem of joint user association and resource allocation in coded cache-enabled HetNets is investigated, considering reducing backhaul transmission through wireless caching.

Some researchers have different optimization objectives. The authors in \cite{10254460} propose a deep reinforcement learning algorithm to optimize user association in a distributed way and maximize the minimum data rate of UEs without global channel state information. The authors in \cite{10264114} propose a dynamic hierarchical game approach for user association and resource allocation in HetNets, considering the scenario where a macro base station provide backhaul services to small cell base stations.
 The authors in \cite{noorivatan2020joint} investigate the problem of joint user association and power-bandwidth allocation in HetNets with the objective of minimizing the mean delay. In \cite{sokun2017novel}, the problem of joint user association and resource allocation allowing opportunistic resource block reuse is investigated, considering two scenarios where user association and power allocation are time-shared or not. The authors in \cite{chaieb2020joint} investigate the problem of joint user association and sub-channel assignment in the scenario where orthogonal and non-orthogonal multiple access techniques coexist, which is solved by a heuristic algorithm with polynomial time complexity.  In \cite{zhao2019deep}, a distributed multi-agent deep reinforcement learning method is proposed for joint user association and resource allocation in HetNets, utilizing dueling double deep-Q network. 

\begin{table*}[htbp]

\centering

\caption{Comparison of Related Literature on Joint User Association and Resource Allocation in Downlink HetNets.}

\label{RelatedLiterature}

\begin{tabular}{c|c|c|c|c|c}

  \hline
   \multirow{2}{*}{\textrm{Reference}} & \multirow{2}{*}{\textrm{Optimization objective}} & \textrm{One macro} & \textrm{Cross-tier} & $\textrm{Single transmitting}$ & $\textrm{Single}$\\
    & &\textrm{base station} & $\textrm{interference}$ & $\textrm{antenna}$ & $\textrm{association}$\\
  \hline
   \cite{wang2017joint} & $\textrm{maximize the weighted sum rate}$ & \ding{55} & \ding{51} & \ding{51} & \ding{55}\\
  \hline
  \cite{chen2015joint} & $\textrm{maximize sum rate}$ & \ding{55} & \ding{51} & \ding{51} & \ding{51}\\
  \hline
  \cite{10211332} & $\textrm{maximize throughput}$ & \ding{55} & \ding{55} & \ding{51} & \ding{51}\\
  \hline
    \cite{cheng2020joint} & $\textrm{maximize the total transmission capacity}$ & \ding{51} & \ding{51} & \ding{55} & \ding{51}\\
  \hline
    \cite{10184114} & $\textrm{maximize sum rate}$ & \ding{51} & \ding{55} & \ding{55} & \ding{51}\\
  \hline
    \cite{han2016backhaul} & $\textrm{maximize the sum of logarithmic rates}$ & \ding{51} & \ding{51} & \ding{51} & \ding{51}\\
  \hline
   \cite{vaezpour2019robust} & $\textrm{maximize the sum of logarithmic rates}$ & \ding{55} & \ding{55} & \ding{51} & \ding{55}\\
  \hline
   \cite{liu2019joint} & $\textrm{maximize the sum of weighted logarithmic rates}$ & \ding{55} & \ding{51} & \ding{51} & \ding{51}\\
  \hline
 \cite{somesula2022artificial} & $\textrm{maximize the energy efficiency}$ & \ding{55} & \ding{51} & \ding{51} & \ding{55}\\
  \hline
    \cite{liu2020joint} & $\textrm{maximize spectrum efficiency and energy efficiency}$ & \ding{51} & \ding{51} & \ding{55} & \ding{55}\\
  \hline
    \cite{MUGHEES2023102206} & $\textrm{maximize energy efficiency}$ & \ding{55} & \ding{51} & \ding{55} & \ding{51}\\
  \hline
  \cite{10246069} & $\textrm{maximize energy efficiency}$ & \ding{51} & \ding{51} & \ding{51} & \ding{51}\\
  \hline
    \cite{li2016joint} & $\textrm{minimize the average energy consumption}$ & \ding{55} & \ding{51} & \ding{51} & \ding{51}\\
  \hline
  \cite{yin2019energy} & $\textrm{minimize total power consumption}$ & \ding{51} & \ding{55} & \ding{51} & \ding{55}\\
  \hline
     \cite{10254460} & $\textrm{maximize the minimum data rate of UEs}$ & \ding{51} & \ding{51} & \ding{51} &\ding{51}\\
  \hline
  \cite{10264114} & $\textrm{optimize wireless backhaul bandwidth allocation and access throughput}$ & \ding{51} & \ding{55} & \ding{51} & \ding{51}\\
  \hline

  \cite{noorivatan2020joint} & $\textrm{minimize the mean delay}$ & \ding{55} & \ding{55} & \ding{51} & \ding{55}\\
  \hline
  \multirow{2}{*}{\cite{sokun2017novel}} & $\textrm{maximize the number of accommodated UEs}$ & \multirow{2}{*}{\ding{55}} & \multirow{2}{*}{\ding{51}} & \multirow{2}{*}{\ding{51}} & \multirow{2}{*}{\ding{51}}\\
  & $\textrm{and minimize the usage of resource blocks}$ & & &\\
  \hline
  \cite{chaieb2020joint} & $\textrm{maximize the total number of associated users}$ & \ding{55} & \ding{55} & \ding{55} & \ding{51}\\
  \hline
  \cite{zhao2019deep} & $\textrm{maximize the long-term downlink utility}$ & \ding{55} & \ding{51} & \ding{51} & \ding{51}\\
  \hline
  
\end{tabular}

\end{table*}

\subsection{Contributions}
In this paper, the problem of application-specific objective-based joint user association and resource allocation for tailored QoS provisioning in downlink 6G HetNets is investigated. It is necessary to point out that although we focus on the downlink transmission scenario of 6G HetNets in this paper, the core idea of the joint user association and resource allocation technique proposed in this paper could be transplanted to the uplink transmission scenario of 6G HetNets and 4G/5G HetNets as well. The optimization objective is set to be optimizing the sum of UEs' application-specific objectives. We formulate this problem as a mixed-integer non-linear programming (MINLP) problem. This problem is inherently difficult to solve directly due to its extremely large solution space caused by its intrinsic NP-hardness and its intrinsic non-convexity caused by the combination of discrete and continuous variables. Therefore,
we decompose the original problem into two subproblems, which are user association and resource allocation, respectively. For resource allocation, both PRB and transmit power are considered. We further decompose the subproblem of resource allocation into two subsubproblems, which are PRB and power allocation with fixed bit error rate (BER), and remaining power allocation.
The contributions of this paper are summarized as follows. 

\begin{enumerate}
   \item A novel concept of tailored QoS provisioning in downlink 6G HetNets is first conceptualized to maximize the summation of the application-specific objectives of all the UEs to ensure the fairness among different UEs with different applications, involving two different resource management tasks, i.e. user association and resource allocation. Moreover, the problem of joint user association and resource allocation for tailored QoS provisioning in downlink 6G HetNets is formulated as a mixed-integer non-linear programming (MINLP) problem.

    \item Due to its intrinsic NP-hardness and non-convexity, the formulated problem is difficult to solve directly. Hence, the original problem is decomposed into two subproblems, i.e. user association and resource allocation. Moreover, an interactive optimization algorithm (IOA) is proposed for solving the subproblems of user association and resource allocation iteratively in an interactive way until convergence is achieved.
    
    \item For PRB and power allocation with fixed BER, an online deferred acceptance (ODA) algorithm based on matching theory is proposed to allocate PRBs with power which could meet the minimum BER requirements of UEs associated with a base station, which is guaranteed to generate a stable solution. In addition, two matching reformulation algorithms are proposed for generating a new stable solution when a new UE is associated with a base station or an associated UE is reassociated with another base station, which could greatly reduce time complexity. For remaining power allocation, a maximum marginal utility descent (MMUD) algorithm is proposed, which divides the remaining transmit power of base station into small pieces and allocate each piece of power to the UE with the maximum marginal utility.

    \item In order to better conduct performance evaluation, extensive experiments with different settings in terms of PBS density and maximum transmit power per PBS are conducted. Simulation results confirm that our proposed IOA algorithm achieves better performance compared with several baseline algorithms in terms of both average utility and UE satisfaction ratio.
\end{enumerate}

\section{System Model}
\label{systemmodelIII}
In this paper, we consider a two-tier HetNet which consists of $M$ macro base stations (MBSs) and $N$ pico base stations (PBSs), denoted by $\mathcal{F} = \{1, \cdots, M, M + 1, \cdots, M + N\}$, where the indexes $1$ to $M$ represent the MBSs and the indexes $M + 1$ to $M + N$ represent the PBSs.

We focus on the downlink transmission scenario of this HetNet and assume that all the base stations employ orthogonal frequency division multiple access (OFDMA). Moreover, we assume that the MBSs are distributed in a planned cellular manner and the PBSs are distributed randomly. We assume that the radii of MBSs and PBSs are denoted by $r_M$ and $r_P$, respectively. Moreover, we assume that the coverage areas of PBSs don't overlap with each other. The coverage area of base station $j$ is denoted by $A_j$. Besides, we denote the maximum transmit power of an MBS and a PBS as $P_{MBS}^{max}$ and $P_{PBS}^{max}$, respectively. And the maximum transmit power of base station $j$ is denoted as $P_j^{max}$.

In addition, we assume that the MBSs and PBSs operate on different frequency bands, since that multiple frequency bands are available in 6G networks which are suitable for different communication scenarios. And using the same frequency band will inevitably bring severe inter-tier interference and degrade the overall network performance. Hence, only intra-tier interference is considered in this paper. Besides, in order to avoid severe intra-tier interference between adjacent MBSs, we assume that reuse-3 scheme \cite{7266476} is employed by MBSs, which is a conventional frequency planning scheme with frequency reuse factor equals 3. Hence, there exist four different frequency bands, denoted by $\mathcal{H} = \{1, 2, 3, 4\}$, and one base station is operated on one frequency band. For base station $j$, we assume that its frequency band is denoted by $h_j$.

For each base station, we assume that its frequency band is divided into $B$ PRBs with equal bandwidth $W$, denoted by $\mathcal{B} = \{1, \cdots, B\}$. We assume that there exist $K$ UEs in the network with heterogeneous QoS requirements, denoted by $\mathcal{K} = \{1, \cdots, K\}$. In addition, we assume that the UEs are distributed in a horizontal area and the coordinate of UE $k$ is denoted by $\phi_k$.

We assume that user association and resource allocation are performed during each time period with length $T$. Moreover, we assume that during each time period, each UE could only be associated with one base station. The user association policy is demonstrated by user association matrix $X$, which is a $K$-by-$(M + N)$ matrix and expressed as:

\begin{equation}
    x_{k,j} = \left \{
    \begin{array}{lr}
    1,\; \textrm{if UE $k$ is associated with base station $j$} \\
    0, \; \textrm{otherwise}
    \end{array}
    \right .
\end{equation}
Hence, the number of UEs associated with base station $j$ is:

\begin{equation}
    K_j = \sum_{k = 1}^K x_{k,j}.
\end{equation}
Moreover, we assume that the set of UEs associated with base station $j$ is denoted by $\mathcal{K}_j = \{k | x_{k,j} = 1, k \in \mathcal{K}\}$.

For each base station, we assume that each PRB could only be allocated to one UE during each time period. Let matrix $\rho_j$ denote the PRB assignment policy of base station $j$, which is a $K$-by-$B$ matrix, expressed as:

\begin{equation}
    \rho_j^{k,b} = \left \{
    \begin{array}{lr}
    1,\; \textrm{if the $b$-th PRB is allocated to UE $k$} \\
    0, \; \textrm{otherwise}
    \end{array}
    \right .
\end{equation}
The data transmission rate of UE $k \in \mathcal{K}$ over PRB $b \in \mathcal{B}$ of base station $j$ can be modelled as follows:

\begin{equation}
    R_{k,j}^b = x_{k,j}\rho_j^{k,b}W\mathrm{log_2}(1+\gamma_{k,j}^b),
\end{equation}
where $\gamma_{k,j}^b$ denotes the signal-to-interference-plus-noise-ratio (SINR), which is:

\begin{equation}
    \gamma_{k,j}^b = \frac{P_j^bg_{k,j}^b}{\sum_{i\in\mathcal{F},i\neq j \land h_i = h_j}P_i^bg_{k,i}^b+\sigma^2},
\end{equation}
where $P_j^b$ is the transmit power of base station $j$ on PRB $b$, $g_{k,j}^b$ is the channel gain accounting for path loss, shadowing and fading from base station $j$ to UE $k$ on PRB $b$. Moreover, $\sigma^2$ denotes the noise power.
Hence, the data rate of UE $k$ is:
\begin{equation}
    R_k = \sum_{j=1}^{M+N}\sum_{p=1}^PR_{k,j}^b.
\end{equation}

We assume that the packet size of UE $k$ is denoted by $s_k$, then the transmission latency of UE $k$ during downlink transmission is:

\begin{equation}
    L_k^{tran} = \frac{s_k}{R_k}.
\end{equation}
We assume that the average latency of the $k$-th UE from the server to its associated base station is denoted by $L_k^{sb}$. The packet arrival of UE $k$ is assumed to follow Poisson process with mean arrival rate $\lambda_k$ packets per transmission time interval (TTI). The length of one TTI is set to be 1 ms. The latency caused by baseband signal processing is neglected as it is not a dominating component. 

Following the modelling of average queuing latency introduced in \cite{9735275}, the queuing of packets of each UE is formulated as an M/G/1 queuing model. Moreover, the average queuing latency of UE $k$ is modelled as:

\begin{equation}
    L_k^{queue} = \frac{\lambda_ks_k^2}{2R_k(R_k - \lambda_ks_k)}
\end{equation}

 Moreover, we assume that the average propagation latency from the associated base station to the $k$-th UE is denoted by $L_k^{prop}$. Hence, the average latency of UE $k$ is:

\begin{equation}
    L_k = L_k^{sb} + L_k^{queue} + L_k^{tran} + L_k^{prop}.
\end{equation}

We assume that the signals are modulated utilizing QPSK, then the bit error rate of UE $k$ over PRB $b \in \mathcal{B}$ of base station $j$ can be expressed as follows:

\begin{equation}
    BER_{k,j}^b = \frac{1}{2}x_{k,j}\mathrm{erfc}(\sqrt{\frac{\gamma_{k,j}^b}{\frac{R_{k,j}^b}{W}}})
    = \frac{1}{2}x_{k,j}\mathrm{erfc}(\sqrt{\frac{W\gamma_{k,j}^b}{R_{k,j}^b}}), \notag
\end{equation}
where $\mathrm{erfc}(x) = \frac{2}{\sqrt{\pi}}\int_x^\infty e^{-\eta^2}d\eta$.
Then the average bit error rate of UE $k$ over all its allocated PRBs is:

\begin{equation}
    BER_k = \frac{\sum_{j=1}^{M+N}\sum_{b=1}^BR_{k,j}^bBER_{k,j}^b}{R_k}.
\end{equation}

We assume that each UE $k \in \mathcal{K}$ has its customized QoS requirements on data transmission rate, latency, and bit error rate, denoted by $R_k^{req}$, $L_k^{req}$ and $BER_k^{req}$, respectively. Moreover, an application-specific objective function is defined for each UE. The application-specific objective function of UE $k$ is expressed as:

\begin{equation}
    U_k = w_k^1\sigma(R_k - R_k^{req}) + w_k^2\sigma(L_k^{req} - L_k),
\end{equation}
where $\sigma(x)$ denotes the sigmoid function, expressed as:

\begin{equation}
    \sigma(x) = \frac{1}{1+e^{-x}}.
\end{equation}
Moreover, we suppose that $w_k^1 + w_k^2 = 1$. Hence, since that $0 < \sigma(x) < 1$, $0 < U_k < 1$.

\section{Problem Formulation}
\label{problemFormulation}
In this paper, we set the optimization objective of the problem of joint user association and resource allocation for tailored QoS provisioning in downlink 6G HetNets as the summation of the application-specific objectives of all the UEs to ensure the fairness among different UEs with different applications.
The joint user association and resource allocation problem is formulated as follows:

\begin{subequations}
\begin{flalign}
& \quad \quad \quad \quad \quad \textrm{(P0)} \; \textrm{Maximize} \quad \sum_{k=1}^{K}U_k\\
 &  \textrm{subject}  \textrm{ to:} \notag\\
    &\sum_{b=1}^BP_j^b \le P_j^{max}, \forall j \in \mathcal{F}\\
    & \sum_{j=1}^{M+N}x_{k,j} = 1, \forall k \in \mathcal{K}, j \in \mathcal{F}\\
    &\sum_{k=1}^{K_j}\sum_{b=1}^B\rho_j^{k,b} \le B, \forall j \in \mathcal{F}\\
    &\rho_j^{k,b}BER_{k, j}^b \le \rho_j^{k,b}BER_k^{req}, \forall j \in \mathcal{F}, \forall k \in \mathcal{K}_j, b \in \mathcal{B} \\
    & x_{k,j} \in \{0,1\}, \forall k \in \mathcal{K}, j \in \mathcal{F} \\
    &\rho_j^{k,b} \in \{0,1\}, \forall k \in \mathcal{K}_j, b \in \mathcal{B}, j \in \mathcal{F}
\end{flalign}
\end{subequations}

In the problem formulation, (13b) constrains that for each base station, the sum of its transmit power allocated to all its PRBs cannot exceed its maximum transmit power constraint. (13c) and (13f) constrain that one UE could be and at most be associated with one base station. (13d) and (13g) constrain that each PRB of each base station could be allocated to at most one UE and for each base station, the sum of PRBs allocated to its UEs cannot exceed its total number of PRBs. (13e) constrains that for any base station, if a PRB is allocated, its BER must satisfy the BER requirement of its corresponding UE. It needs to be pointed out that due to the constrained amounts of available resources of base stations, when the number of UEs is very large, it might be impossible to satisfy each UE's QoS requirements. Hence, we don't constrain the generated solution must satisfy each UE's QoS requirements, in case that no such valid solution could be generated. 

Once user association policy (user association matrix $X$) is generated, the resource allocation problem needs to be solved for each base station. The resource allocation problem of base station $j$ is formulated as follows:

\begin{subequations}
\begin{flalign}
\quad & \quad \textrm{(P1)} \; \textrm{Maximize} \quad \sum_{k \in \mathcal{K}_j}U_k\\
   \textrm{subject} & \textrm{ to:}  \notag\\
    &\sum_{b=1}^BP_j^b \le P_j^{max}\\
    & \sum_{k=1}^{K_j}\sum_{b=1}^B\rho_j^{k,b} \le B\\
    &\rho_j^{k,b}BER_{k, j}^b \le \rho_j^{k,b}BER_k^{req}, \forall k \in \mathcal{K}_j, b \in \mathcal{B}\\
    &\rho_j^{k,b} \in \{0,1\}, \forall k \in \mathcal{K}_j, b \in \mathcal{B}
\end{flalign}
\end{subequations}

\section{Decomposition of Resource Allocation Problem}
From Section \ref{problemFormulation}, we can find that problem P0 is a mixed-integer non-linear programming (MINLP) problem. Hence, it is NP-hard and non-convex. Therefore, in order to solve problem P0, an interactive optimization algorithm is proposed in this paper, whose components and the whole algorithm will be introduced in the following sections. In this section, resource allocation in the initialization stage will be introduced.

According to (14d), if a PRB is allocated, its BER must satisfy the BER requirement of its corresponding UE. Hence, to generate valid solutions, for resource allocation, we decompose it into two subsubproblems, which are PRB and power allocation with fixed BER and remaining power allocation. 
The subsubproblem of PRB and power allocation with fixed BER is formulated as shown below.

\begin{subequations}
\begin{flalign}
\quad & \quad \textrm{(P2)} \; \textrm{Maximize} \sum_{k \in \mathcal{K}_j}U_k\\
\textrm{subject}  & \textrm{ to:}  \notag\\
    &\sum_{b=1}^BP_j^b \le P_j^{max}\\
    & \sum_{k=1}^{K_j}\sum_{b=1}^B\rho_j^{k,b} \le B\\
    &\rho_j^{k,b}BER_{k, j}^b = \rho_j^{k,b}BER_k^{req}, \forall k \in \mathcal{K}_j, b \in \mathcal{B}\\
    & \rho_j^{k,b} \in \{0,1\}, \forall k \in \mathcal{K}_j, b \in \mathcal{B}
\end{flalign}
\end{subequations}

We can find that the difference between P2 and P1 is that in P2, it's constrained that if one PRB is allocated to an UE, its BER must be equal to the BER requirement of its corresponding UE. Once the subproblem P2 is solved, if there exists remaining power, it needs to be allocated to the UEs for maximizing their utilities. The subsubproblem of remaining power allocation is formulated as follows.

\begin{subequations}
\begin{flalign}
 \textrm{(P3)} & \;  \textrm{Maximize} \sum_{k \in \mathcal{K}_j}U_k\\
\textrm{subject}   \textrm{ to:}  & \notag\\
    & \sum_{b=1}^B\Delta P_j^b \le P_j^{rem} \quad \quad \quad \quad \quad \quad \quad \quad \quad 
\end{flalign}
\end{subequations}
where $\Delta P_j^b$ denotes the amount of remaining transmit power allocated to the $b$-th PRB of the $j$-th base station, and $P_j^{rem}$ denotes the total amount of remaining transmit power of base station $j$.

\section{User Association in the Initialization Stage}
In the initialization stage, in order to fully utilize the resources of pico base stations, the UEs located within the coverage areas of PBSs are associated with the corresponding PBSs. Moreover, the remaining UEs are associated with corresponding MBSs. The algorithm of user association in the initialization stage is shown in Alg. \ref{UAInitialization}.

\begin{algorithm}[hbtp]
\caption{User Association in the Initialization Stage}
\label{UAInitialization}
\begin{algorithmic}
\REQUIRE ~~
The coverage areas of base stations, $A_j$, $1 \le j \le M+N$;\\
The location of UEs, $\phi_k$, $1 \le k \le K$;\\
\ENSURE ~~
The user association matrix, $X$;\\
\STATE $X = O_{K\times(M+N)}$;\\
\FOR{$k \in \mathcal{K}$}
\STATE FLAG $=$ \FALSE;\\
\FOR{$M+1 \le j \le M+N$}
\IF{$\phi_k \in A_j$}
\STATE $x_{k,j} = 1$;\\
\STATE FLAG $=$ \TRUE;\\
\STATE \textbf{Break};\\
\ENDIF
\ENDFOR
\IF{FLAG $=$ \FALSE}
\FOR{$1 \le j \le M$}
\IF{$\phi_k \in A_j$}
\STATE $x_{k,j} = 1$;\\
\STATE \textbf{Break};\\
\ENDIF
\ENDFOR
\ENDIF
\ENDFOR
\end{algorithmic}
\end{algorithm}

\section{Resource Allocation in the Initialization Stage} 
In this section, the proposed solutions for resource allocation in the initialization stage will be introduced.

\subsection{PRB and Power Allocation with Fixed BER}
The subsubproblem of PRB and power allocation with fixed BER could be transformed into a 
many-to-one matching game with externality. To formulate the many-to-one matching game, we make some definitions as shown below.

\begin{myDef}
\label{def1}
Given two disjoint sets $\mathcal{B}$ for PRBs of base station $j$ and $\mathcal{K}_j$ for UEs associated with base station $j$, the mapping from PRB to UE is defined as a many-to-one mapping $\mu: \mathcal{B} \rightarrow \mathcal{K}_j$ that satisfies:

1) $\mu(b) \in \mathcal{K}_j \cup \{\emptyset\}, \forall b \in \mathcal{B}$,

2) $|\mu(b)| \in \{0, 1\}, \forall b \in \mathcal{B}$,

3) $\mu^{-1}(k) \in \mathcal{B} \cup \{\emptyset\}, \forall k \in \mathcal{K}_j$,

\noindent where $|\mu(\cdot)|$ is the cardinality of the matching result $\mu(\cdot)$.
\end{myDef}

According to Definition \ref{def1}, we can find that different from conventional many-to-one matching, in this matching, there exists no constraint on the maximum capacity of UEs.  

From Section \ref{systemmodelIII}, we know that the utility, data rate and latency of each UE are totally determined by the set of PRBs allocated to it. Hence, we denote the utility of UE $k$ as $U_k^{\mathcal{B}_k}$, where $\mathcal{B}_k$ denotes the set of PRBs allocated to it. Moreover, we denote the data rate and latency of UE $k$ as $R_k^{\mathcal{B}_k}$ and $L_k^{\mathcal{B}_k}$, respectively. Moreover, let $\delta_k^{\mathcal{B}_k}$ indicate whether the QoS requirements of UE $k$ are satisfied,

\begin{equation}
    \delta_k^{\mathcal{B}_k} = \left \{
    \begin{array}{cc}
        0, & \textrm{if} \; R_k^{\mathcal{B}_k} < R_k^{req} \; \textrm{or} \; L_k^{\mathcal{B}_k} > L_k^{req}  \\
        1, & \textrm{if} \; R_k^{\mathcal{B}_k} \ge R_k^{req} \; \textrm{and} \; L_k^{\mathcal{B}_k} \le L_k^{req}
    \end{array}
    \right .
\end{equation}

\begin{myDef}
\label{def2}
The preference extent of PRB $b$ to UE $k$ is defined as:
\begin{equation}
PE_b^k = \left \{ 
\begin{array}{cc}
     U_k^{\mathcal{B}_k \cup \{b\}} - U_k^{\mathcal{B}_k}, & \textrm{if} \; b \notin \mathcal{B}_k \; \textrm{and} \; \delta_k^{\mathcal{B}_k} = 1\\ 
     U_k^{\mathcal{B}_k} - U_k^{\mathcal{B}_k - \{b\}}, & \textrm{if} \; b \in \mathcal{B}_k \; \textrm{and} \; \delta_k^{\mathcal{B}_k - \{b\}} = 1\\
     2 - U_k^{\mathcal{B}_k}, & \textrm{if} \; b \notin \mathcal{B}_k \; \textrm{and} \;\delta_k^{\mathcal{B}_k} = 0\\
     2 - U_k^{\mathcal{B}_k - \{b\}}, & \textrm{if} \; b \in \mathcal{B}_k \; \textrm{and} \;\delta_k^{\mathcal{B}_k - \{b\}} = 0 \notag
\end{array}
\right .
\end{equation}
\end{myDef}

From Definition \ref{def2}, we can find that if $b \notin \mathcal{B}_k \; \textrm{and} \; \delta_k^{\mathcal{B}_k} = 1$ or $b \in \mathcal{B}_k \; \textrm{and} \; \delta_k^{\mathcal{B}_k - \{b\}} = 1$, $0 < PE_b^k < 1$. And we can find that if $b \notin \mathcal{B}_k \; \textrm{and} \; \delta_k^{\mathcal{B}_k} = 0$ or $b \in \mathcal{B}_k \; \textrm{and} \; \delta_k^{\mathcal{B}_k - \{b\}} = 0$, $PE_b^k > 1$. This ensures that the UEs whose QoS requirements has not been satisfied and the UEs whose QoS requirements will be unsatisfied if PRB $b$ is removed from its PRB set have higher priorities than those UEs whose QoS requirements are satisfied no matter how PRB $b$ is allocated.

\begin{myDef}
\label{def3}
The preference extent of UE $k$ to PRB $b$ is defined as:
\begin{equation}
PE_k^b = \left \{
\begin{array}{cc}
\frac{U_k^{\mathcal{B}_k \cup \{b\}} - U_k^{\mathcal{B}_k}}{\theta_b^k}, \textrm{if} \; b \notin \mathcal{B}_k\\
\frac{U_k^{\mathcal{B}_k} - U_k^{\mathcal{B}_k - \{b\}}}{\theta_b^k}, \textrm{if} \; b \in \mathcal{B}_k
\end{array}
\right .
\end{equation}
where $\theta_b^k$ denotes the amount of power required by PRB $b$ to satisfy the BER requirement of UE $k$.
\end{myDef} 

\begin{myDef}
\label{def4}
PRB $b$ prefers UE $k_1$ to UE $k_2$, if $PE_b^{k_1} > PE_b^{k_2}$, denoted by $k_1 \succ_b k_2$, for $b \in \mathcal{B}$, $k_1, k_2 \in \mathcal{K}_j$, $k_1 \neq k_2$.
\end{myDef}

\begin{myDef}
\label{def5}
UE $k$ prefers PRB $b_1$ to PRB $b_2$, if $PE_k^{b_1} > PE_k^{b_2}$, denoted by $b_1 \succ_k b_2$, for $k \in \mathcal{K}_j$, $b_1, b_2 \in \mathcal{B}$, $b_1 \neq b_2$.
\end{myDef}

\begin{myDef}
\label{def6}
In matching $\mu$, a PRB-UE pair $(b, k)$ is a blocking pair if $k \succ_b \mu(b)$, where $\mu(b)$ represents the UE which PRB $b$ is allocated to.
\end{myDef}
\begin{myDef}
\label{def7}
A matching $\mu$ is stable, if there exists no blocking pair in $\mu$.
\end{myDef}

\begin{algorithm}[hbtp]
\caption{Online Deferred Acceptance Algorithm}
\label{ODA}
\begin{algorithmic}
\STATE 1. Initialization:\\
\STATE Initiate the set of PRBs which has not been allocated to UEs, $\mathcal{B'} = \mathcal{B}$;\\
\FOR{$k \in \mathcal{K}_j$}
\STATE Initiate the set of PRBs allocated to UE $k$, $\mathcal{B}_k = \emptyset$;
\ENDFOR
\FOR{$b \in \mathcal{B'}$}
\FOR{$k \in \mathcal{K}_j$}
\STATE Calculate the amount of power required for PRB $b$ to meet the BER requirement of UE $k$, $\theta_b^k$;\\
\ENDFOR
\ENDFOR

\STATE 2. Matching: \\
\WHILE{$\mathcal{B'} \neq \emptyset$}
\FOR{$b \in \mathcal{B'}$}
\FOR{$k \in \mathcal{K}_j$}
\STATE UE $k$ sends PRB $b$ its utility gain $\Delta U_k^b = U_k^{\mathcal{B}_k \cup \{b\}} - U_k^{\mathcal{B}_k}$ if PRB $b$ is allocated to UE $k$, and the indicator of its QoS requirement satisfaction, $\delta_k^{\mathcal{B}_k}$;\\
\ENDFOR
\STATE PRB $b$ calculates its preference extents on the UEs in set $\mathcal{K}_j$ and randomly selects one UE $k^*$ with the maximum preference extent.
\STATE PRB $b$ applies for UE $k^*$;
\ENDFOR

\FOR{$k \in \mathcal{K}_j$}
\STATE UE $k$ calculates its preference extents on the PRBs applied for it and randomly selects one PRB $b^*$ with the maximum preference extent;
\STATE $\mathcal{B}_k = \mathcal{B}_k \cup \{b^*\}$;
\STATE $\mathcal{B'} = \mathcal{B'} - \{b^*\}$;
\ENDFOR

\ENDWHILE
\STATE 3. Obtain the matching result of PRBs and UEs, $\mu^*$;
\end{algorithmic}
\end{algorithm}

For solving the formulated many-to-one matching game, the Online Deferred Acceptance (ODA) algorithm is proposed, as shown in Alg. \ref{ODA}, whose basic idea is shown as follows:

1. During each iteration, each PRB applies for the UE with the maximum preference extent and the UEs whose QoS requirements haven't been satisfied are given the highest priority;

2. During each iteration, each UE accepts the PRB applied to it with the best channel condition;

3. To deal with the impact of externality, the preference extents of PRBs and UEs are calculated in an online way.

\begin{theorem}
The matching $\mu^*$ resulting from Alg. \ref{ODA} is stable.
\end{theorem}

\begin{proof}
\textbf{(Proof by contradiction)} Assume that there exists one blocking pair $(b, k)$ in the matching result $\mu^*$. Then according to Definition \ref{def6}, we have: $k \succ_b \mu^*(b)$. Then according to Definition \ref{def4}, we have: $PE_b^{k} > PE_b^{\mu^*(b)}$. In order to facilitate the process of proof, we denote the preference extent of PRB $b$ to UE $k$ during iteration $t$ as $PE_b^k(t)$ and denote the set of PRBs allocated to UE $k$ after iteration $t$ as $\mathcal{B}_k(t)$. Suppose that the last PRB allocated to UE $\mu^*(b)$ is PRB $b'$ and it was allocated to UE $\mu^*(b)$ during iteration $t$. Hence, during iteration $t$, $PE_{b'}^{\mu^{*}(b)}(t) \ge PE_{b'}^{k}(t)$. And according to Definition \ref{def2}, we know that $PE_{b'}^{\mu^{*}(b)} = PE_{b'}^{\mu^{*}(b)}(t)$, which are both the utility gain introduced by allocating PRB $b'$ to UE ${\mu^{*}(b)}$. Moreover, because each PRB is allocated the amount of power to just satisfy the BER requirement of its corresponding UE, deleting any PRB from $\mathcal{B}_{\mu^{*}(b)}$ will result in the same amount of utility loss. And allocating any PRB in set $\mathcal{B} - \mathcal{B}_k$ to UE $k$ will result in the same amount of utility gain. Hence, we know $PE_{b'}^{\mu^{*}(b)} = PE_{b}^{\mu^{*}(b)}$ and $PE_{b'}^k = PE_{b}^k$. In addition, due to the property of Sigmoid function, we know $PE_{b'}^{k}(t) \ge PE_{b'}^{k}$, because $|\mathcal{B}_k| \ge |\mathcal{B}_k(t)|$. Hence, we have $PE_b^{\mu^*(b)} \ge PE_b^{k}$, which is a contradiction.

In conclusion, there doesn't exist any blocking pair in $\mu^{*}$ and $\mu^{*}$ is stable.
\end{proof}

\begin{property}
The time complexity of Alg. \ref{ODA} is $\mathcal{O}(K_jB^2)$.
\end{property}

\begin{proof}
In Alg. \ref{ODA}, there exists two procedures at each iteration in the matching phase, which are application and allocation. In the application procedure, the most time-consuming step is selecting one UE with the maximum preference extent for each PRB. This step requires $(K_j - 1)$ comparison operations. In the worst case, in each iteration, only one PRB will be allocated. And the total number of required comparison operations in the application procedure is: $\sum_{i=1}^B(K_j - 1) = \frac{(K_j - 1)B(B+1)}{2}$. In the allocation procedure, each UE selects one PRB with the maximum preference extent from the PRBs which applied to it. In each iteration, the number of required comparison operations equals the number of unallocated PRBs minus the number of UEs which have at least one PRB candidate. In the worst case, in each iteration, all the unallocated PRBs apply to one and the same UE and only one PRB will be allocated. And the total number of required comparison operations in the allocation procedure is: $\sum_{i=1}^B(i - 1) = \frac{B(B-1)}{2}$. Therefore, in the worst case, the total number of required comparison operations of Alg. \ref{ODA} is $\frac{(K_j - 1)B(B+1)}{2} + \frac{B(B-1)}{2}$. Hence, the time complexity of Alg. \ref{ODA} is $\mathcal{O}(K_jB^2)$.
\end{proof}

After the completion of PRB allocation, the required amounts of power are allocated to PRBs for satisfying their BER requirements, as shown in Alg. \ref{PAfixedBER}.

\begin{algorithm}[hbtp]
\caption{Power Allocation with Fixed BER}
\label{PAfixedBER}
\begin{algorithmic}
\REQUIRE ~~
The maximum amount of transmit power of base station $j$, $P_j^{max}$;\\
The matching result of PRBs and UEs associated with base station $j$, $\mu$;\\
The amount of power required for PRB $b \in \mathcal{B}$ to meet the BER requirement of UE $k \in \mathcal{K}_j$, $\theta_b^k$;\\
\ENSURE ~~
The power resource allocation result, $P_j^b$, $b \in \mathcal{B}$;\\
\STATE $P_j = P_j^{max}$;\\
\FOR{$b \in \mathcal{B}$}
\STATE $P_j^b = 0$;\\
\ENDFOR
\STATE Randomly shuffle the set $\mathcal{K}_j$;\\
\FOR{$k \in \mathcal{K}_j$}
\STATE $FLAG = \FALSE$;\\
\FOR{$b \in \mu(k)$}
\IF{$P_j \ge \theta_b^k$}
\STATE $P_j = P_j - \theta_b^k$;\\
\STATE $P_j^b = \theta_b^k$;\\
\ELSE
\STATE $FLAG = \TRUE$;
\STATE \textbf{Break};\\
\ENDIF
\ENDFOR
\IF{$FLAG = \TRUE$}
\STATE \textbf{Break};\\
\ENDIF
\ENDFOR
\end{algorithmic}
\end{algorithm}

\subsection{Remaining Power Allocation}
After solving the subsubproblem P2, if there exists remaining power of the base station, the subsubproblem P3 needs to be solved. In order to solve P3, a Maximum Marginal Utility Descent (MMUD) algorithm is proposed, as shown in Alg. \ref{MMUD}. The basic idea of MMUD algorithm is to divide the remaining power resource into small pieces and allocate each piece of power resource to the PRB with the maximum marginal utility corresponding to the amount of power resource allocated to it at each iteration. Since marginal utility reflects the demand degree on power resource, giving priority to the PRB with the highest marginal utility could result in greater increase of the value of optimization objective. Based on Alg. \ref{PAfixedBER} and \ref{MMUD}, the power allocation algorithm is shown in Alg. \ref{PRA}.

\begin{algorithm}[hbtp]
\caption{Maximum Marginal Utility Descent Algorithm}
\label{MMUD}
\begin{algorithmic}
\STATE 1. Initialization:\\
\STATE Divide the remaining transmit power of base station $j$ into $n$ pieces, each with amount $\alpha = \frac{P^{rem}_j}{n}$;\\
\STATE $i = 1$;\\
\FOR{$k \in \mathcal{K}_j$}
\STATE Calculate the marginal utility corresponding to the amount of power resource allocated to PRB $b$: $\frac{\partial U_k}{\partial P_j^b}$, where $\rho_j^{k,b} = 1$.
\ENDFOR
\STATE 2. Allocation: \\
\FOR{$i \le n$}
\STATE Select one PRB $b^*$ with the maximum marginal utility corresponding to the power resource allocated to it;\\
\STATE $P_j^{b^*} = P_j^{b^*} + \alpha$;\\
\STATE $P^{rem}_j = P^{rem}_j - \alpha$;\\
\STATE Calculate and update the marginal utility corresponding to the amount of power resource allocated to PRB ${b^*}$: $\frac{\partial U_k}{\partial P_j^{b^*}}$, where $\rho_j^{k,{b^*}} = 1$;\\
\STATE $i = i + 1$.\\
\ENDFOR
\end{algorithmic}
\end{algorithm}

\begin{algorithm}[hbtp]
\caption{Power Allocation Algorithm}
\label{PRA}
\begin{algorithmic}
\REQUIRE ~~
The total amount of available power resource of base station $j$, $P_j$;\\
The matching result of PRBs and UEs associated with base station $j$, $\mu$;\\
The amount of power required for PRB $b \in \mathcal{B}$ to meet the BER requirement of UE $k \in \mathcal{K}_j$, $\theta_b^k$;\\
\ENSURE ~~
The power resource allocation result, $P_j^b$, $b \in \mathcal{B}$;\\
\STATE Invoke Alg. \ref{PAfixedBER} to allocate power with fixed BER;\\
\STATE $P_j^{rem} = P_j - \sum_{b \in \mathcal{B}}P_j^b$;\\
\IF{$P_j^{rem} > 0$}
\STATE Invoke Alg. \ref{MMUD} to allocate remaining power to UEs;\\
\ENDIF
\end{algorithmic}
\end{algorithm}

\section{Resource Allocation in the Optimization Stage}
In the optimization stage, user association results and resource allocation results will be changed interactively. Once user association results are changed, resource allocation results need to be changed accordingly. If we use Alg. \ref{ODA} in the initialization stage in each iteration of the optimization stage, the total time complexity will be very high. In order to reduce time complexity, two algorithms are proposed to generate new stable PRB-UE matchings based on the former PRB-UE matchings. Moreover, Alg. \ref{PRA} in the initialization stage is also utilized for power allocation for each base station in the optimization stage.

\subsection{Matching Reformulation with One Deleted UE}
When a UE is removed from the UE set of base station $j$, the PRBs allocated to that UE should be reallocated to other UEs associated with base station $j$. Alg. \ref{md} is proposed to generate a new PRB-UE matching based on the former PRB-UE matching.

\begin{algorithm}[hbtp]
\caption{Matching Reformulation With One Deleted UE}
\label{md}
\begin{algorithmic}
\REQUIRE ~~
The current matching result, $\mu$;\\
The deleted UE, $k'$;\\
\ENSURE ~~
The new matching result, $\mu^*$\\
\STATE 1. Initialization:\\
\STATE Compute the set of PRBs allocated to UE $k'$ in $\mu$, $\mathcal{B}_{k'}$;\\

\FOR{$b \in \mathcal{B}_{k'}$}
\FOR{$k \in \mathcal{K}_j$}
\STATE Calculate the amount of power required for PRB $b$ to meet the BER requirement of UE $k$, $\theta_b^k$;\\

\ENDFOR
\ENDFOR

\STATE 2. Matching:\\
\WHILE{$\mathcal{B}_{k'} \neq \emptyset$}
\FOR{$b \in \mathcal{B}_{k'}$}
\FOR{$k \in \mathcal{K}_j$}
\STATE UE $k$ sends PRB $b$ its utility gain $\Delta U_k^b = U_k^{\mathcal{B}_k \cup \{b\}} - U_k^{\mathcal{B}_k}$ if PRB $b$ is allocated to UE $k$, and the indicator of its QoS requirement satisfaction, $\delta_k^{\mathcal{B}_k}$;\\
\ENDFOR
\STATE PRB $b$ calculates its preference extents on the UEs in set $\mathcal{K}_j$ and randomly selects one UE $k^*$ with the maximum preference extent.
\STATE PRB $b$ applies for UE $k^*$;
\ENDFOR

\FOR{$k \in \mathcal{K}_j$}
\STATE UE $k$ calculates its preference extents on the PRBs applied for it and randomly selects one PRB $b^*$ with the maximum preference extent;
\STATE $\mathcal{B}_k = \mathcal{B}_k \cup \{b^*\}$;
\STATE $\mathcal{B}_{k'} = \mathcal{B}_{k'} - \{b^*\}$;
\ENDFOR
\ENDWHILE
\end{algorithmic}
\end{algorithm}

\begin{theorem}
The matching $\mu^*$ resulting from Alg. \ref{md} is stable as long as the input matching $\mu$ is stable.
\end{theorem}

\begin{proof}
\textbf{(Proof by contradiction)} In order to facilitate the process of proof, we denote the set of PRBs allocated to UE $k$ in the matching $\mu$ as $\mathcal{B}_k^\mu$. Moreover, we denote the preference extent of PRB $b$ to UE $k$ in matching $\mu$ as $PE_b^k|_\mu$. Assume that there exist one blocking pair $(b, k)$ in the matching result $\mu^*$. Then according to Definition \ref{def6}, we have: $k \succ_b \mu^*(b)$. Then according to Definition \ref{def4}, we have: $PE_b^{k}|_{\mu^*} > PE_b^{\mu^*(b)}|_{\mu^*}$. 
There exists three scenarios to be considered. 

\textbf{Scenario 1:} $\mathcal{B}_k^\mu = \mathcal{B}_k^{\mu^*}$ and $\mathcal{B}_{\mu^*(b)}^\mu = \mathcal{B}_{\mu^*(b)}^{\mu^*}$.

In this scenario, there exists one blocking pair $(b, k)$ in matching $\mu$. Hence, matching $\mu$ is unstable, which is a contradiction.

\textbf{Scenario 2:} $\mathcal{B}_k^\mu \neq \mathcal{B}_k^{\mu^*}$ and $\mathcal{B}_{\mu^*(b)}^\mu = \mathcal{B}_{\mu^*(b)}^{\mu^*}$.

In this scenario, we know that additional PRBs are allocated to UE $k$ in matching $\mu^*$. Since matching $\mu$ is stable, we know that $PE_b^k|_\mu \le PE_b^{\mu(b)}|_\mu$. Because additional PRBs are allocated to UE $k$ in matching $\mu^*$, due to the property of Sigmoid function, we know $PE_b^k|_{\mu^*} \le PE_b^k|_\mu$. And since no more PRB is allocated to UE $\mu(b)$ in matching $\mu^*$, $PE_b^{\mu^*(b)}|_{\mu^*} = PE_b^{\mu^*(b)}|_\mu$. Hence, $PE_b^{k}|_{\mu^*} < PE_b^{\mu^*(b)}|_{\mu^*}$, which is a contradiction.

\textbf{Scenario 3:} $\mathcal{B}_{\mu^*(b)}^\mu \neq \mathcal{B}_{\mu^*(b)}^{\mu^*}$.

In this scenario, we know that additional PRBs are allocated to UE $\mu^*(b)$ in matching $\mu^*$. Suppose that the last PRB allocated to UE $\mu^*(b)$ is PRB $b'$. Hence, we know $PE_{b'}^{\mu^*(b)} \ge PE_{b'}^{k}$. Because $PE_{b'}^{\mu^*(b)} = PE_{b}^{k}$ and $PE_{b'}^{\mu^*(b)} = PE_{b}^{k}$, we know that $PE_{b}^{\mu^*(b)} \ge PE_{b}^{k}$, which is a contradiction.

In conclusion, there doesn't exist any blocking pair in matching $\mu^*$ and matching $\mu^*$ is stable.

\end{proof}

\begin{property}
The time complexity of Alg. \ref{md} is $\mathcal{O}(K_j|\mathcal{B}_{k'}|^2$), where $|\mathcal{B}_{k'}|$ denotes the cardinality of set $\mathcal{B}_{k'}$.
\end{property}

\begin{proof}
From Alg. \ref{md}, we know that the most time-consuming part is the matching phase. In the matching phase, there exists two procedures at each iteration in the matching phase, which are application and allocation. In the application procedure, the most time-consuming step is selecting one UE with the maximum preference extent for each PRB. This step requires $(K_j - 1)$ comparison operations. In the worst case, in each iteration, only one PRB will be allocated. And the total number of required comparison operations in the application procedure is: $\sum_{i=1}^{|\mathcal{B}_{k'}|}(K_j - 1) = \frac{(K_j - 1)|\mathcal{B}_{k'}|(|\mathcal{B}_{k'}|+1)}{2}$. In the allocation procedure, each UE selects one PRB with the maximum preference extent from the PRBs which applied to it. In each itertaion, the number of required comparison operations equals the number of unallocated PRBs minus the number of UEs which have at least one PRB candidate. In the worst case, in each iteration, all the unallocated PRBs apply to one and the same UE and only one PRB will be allocated. And the total number of required comparison operations in the allocation procedure is: $\sum_{i=1}^{|\mathcal{B}_{k'}|}(i - 1) = \frac{|\mathcal{B}_{k'}|(|\mathcal{B}_{k'}|-1)}{2}$. Therefore, in the worst case, the total number of required comparison operations of the ODA algorithm is $\frac{(K_j - 1)|\mathcal{B}_{k'}|(|\mathcal{B}_{k'}|+1)}{2} + \frac{|\mathcal{B}_{k'}|(|\mathcal{B}_{k'}|-1)}{2}$. Hence, the time complexity of Alg. \ref{md} is $\mathcal{O}(K_j|\mathcal{B}_{k'}|^2)$.
\end{proof}

\subsection{Matching Reformulation with One Added UE}
When an UE is added to the UE set of base station $j$, some PRBs allocated to the UEs in the original UE set need to be reallocated to this UE. Alg. \ref{ma} is proposed to generate a new PRB-UE matching based on the former PRB-UE matching.

\begin{algorithm}[hbtp]
\caption{Matching Reformulation with One Added UE}
\label{ma}
\begin{algorithmic}
\REQUIRE ~~
The current matching result, $\mu$;\\
The added UE, $k'$;\\
\ENSURE ~~
The new matching result, $\mu^*$\\
\STATE 1. Initialization:\\
\STATE $\mathcal{B}_{k'} = \emptyset$
\STATE 2. Matching:\\
\WHILE{\TRUE}
\STATE $\mathcal{B'} = \emptyset$;
\FOR{$k \in \mathcal{K}_j$}
\STATE Randomly select one PRB $b$ allocated to UE $k$;
\STATE $\mathcal{B'} = \mathcal{B'} \cup \{b\}$;
\STATE Calculate the amount of power required for PRB $b$ to meet the BER requirement of UE $k'$, $\theta_b^{k'}$;
\ENDFOR
\FOR{$b \in \mathcal{B'}$}
\STATE UE $k'$ sends PRB $b$ its utility gain $\Delta U_{k'}^b = U_{k'}^{\mathcal{B}_{k'} \cup \{b\}} - U_{k'}^{\mathcal{B}_{k'}}$ if PRB $b$ is allocated to UE $k'$, and the indicator of its QoS requirement satisfaction, $\delta_{k'}^{\mathcal{B}_{k'}}$;\\
\STATE PRB $b$ computes its preference extents on UE $\mu(b)$ and $k'$, $PE_b^{\mu(b)}$ and $PE_b^{k'}$;\\
\IF{$PE_b^{k'} \le PE_b^{\mu(b)}$}
\STATE $\mathcal{B'} = \mathcal{B'} - \{b\}$;
\ENDIF
\ENDFOR

\IF{$\mathcal{B'} = \emptyset$}
\STATE Obtain the new matching result, $\mu^*$;\\
\STATE \textbf{Break};\\
\ELSE
\STATE UE $k'$ randomly selects one PRB $b^*$ from $\mathcal{B'}$ with the lowest preference extent on its UE;
\STATE $\mathcal{B}_{\mu(b^*)} = \mathcal{B}_{\mu(b^*)} - \{b^*\}$
\STATE $\mathcal{B}_{k'} = \mathcal{B}_{k'} + \{b^*\}$;
\ENDIF
\ENDWHILE
\end{algorithmic}
\end{algorithm}

\begin{theorem}
The matching $\mu^*$ resulting from Alg. \ref{ma} is stable as long as the input matching $\mu$ is stable.
\end{theorem}

\begin{proof}
\textbf{(Proof by contradiction)}
 Assume that there exist one blocking pair $(b, k)$ in the matching result $\mu^*$. Then according to Definition \ref{def6}, we have: $k \succ_b \mu^*(b)$. Then according to Definition \ref{def4}, we have: $PE_b^{k}|_{\mu^*} > PE_b^{\mu^*(b)}|_{\mu^*}$. There exists seven scenarios to be considered.
 
 \textbf{Scenario 1:} $k = k'$.
 
 In this scenario, from the process of Alg. \ref{ma}, we know that the algorithm won't stop if $PE_b^{k'} > PE_b^{\mu^*(b)}$. Hence, we have 
 $PE_b^{k'}|_{\mu^*} \le PE_b^{\mu^*(b)}|_{\mu^*}$, which is a contradiction.
 
 \textbf{Scenario 2:} $k \neq k'$, $\mu^*(b) \neq k'$, $\mathcal{B}_k^\mu = \mathcal{B}_k^{\mu^*}$ and $\mathcal{B}_{\mu^*(b)}^\mu = \mathcal{B}_{\mu^*(b)}^{\mu^*}$.
 
 In this scenario, we know that the PRBs allocated to UE $k$ and $\mu^*(b)$ are unchanged in matching $\mu^*$ compared with those in matching $\mu$. Hence, $PE_b^{k}|_{\mu^*} = PE_b^{k}|_{\mu}$ and $PE_b^{\mu^*(b)}|_{\mu^*} = PE_b^{\mu^*(b)}|_{\mu}$.  Since matching $\mu$ is stable, we know that $PE_b^{k}|_{\mu} \le PE_b^{\mu^*(b)}|_{\mu}$. Hence, $PE_b^{k}|_{\mu^*} \le PE_b^{\mu^*(b)}|_{\mu^*}$, which is a contradiction.
 
 \textbf{Scenario 3:} $k \neq k'$, $\mu^*(b) \neq k'$, $\mathcal{B}_k^\mu = \mathcal{B}_k^{\mu^*}$ and $\mathcal{B}_{\mu^*(b)}^\mu \neq \mathcal{B}_{\mu^*(b)}^{\mu^*}$.
 
 In this case, we know that no PRB of UE $k$ was reallocated to UE $k'$ during the matching process and at least one PRB of UE $\mu^*(b)$ was reallocated to UE $k'$. Hence, due to the property of sigmoid function and Definition 2, we have $PE_b^{\mu^*(b)}|_{\mu^*} \ge PE_b^{\mu^*(b)}|_{\mu}$. Since matching $\mu$ is stable, we know $PE_b^{k}|_{\mu} \le PE_b^{\mu^*(b)}|_{\mu}$. Because $\mathcal{B}_k^\mu = \mathcal{B}_k^{\mu^*}$, we know $PE_b^{k}|_{\mu} = PE_b^{k}|_{\mu^*}$. Hence, $PE_b^{\mu^*(b)}|_{\mu^*} \ge PE_b^{k}|_{\mu^*}$, which is a contradiction.
 
 \textbf{Scenario 4:} $k \neq k'$, $\mu^*(b) \neq k'$, $\mathcal{B}_k^\mu \neq \mathcal{B}_k^{\mu^*}$ and $\mathcal{B}_{\mu^*(b)}^\mu = \mathcal{B}_{\mu^*(b)}^{\mu^*}$.
 
 In this case, no PRB of UE $\mu^*(b)$ was reallocated during the matching process. And at least one PRB of UE $k$ was reallocated during the matching process. Suppose that the last reallocated PRB of UE $k$ was reallocated at time $t$ and the last reallocated PRB of UE $k$ is $b_k$. According to Alg. \ref{ma}, we know that $PE_b^{\mu^*(b)}(t) \ge PE_{b_k}^{k}(t)$. And according to Definition 2, the preference extent of PRB $b_k$ to UE $k$ stays unchanged after it was reallocated. Hence, we have $PE_{b}^{\mu^*(b)}|_{\mu^*} = PE_b^{\mu^*(b)}(t)$ and $PE_{b_k}^{k}|_{\mu^*} = PE_{b_k}^{k}(t)$. Moreover, because $PE_{b}^{k}|_{\mu^*} = PE_{b_k}^{k}|_{\mu^*}$, we know $PE_b^{\mu^*(b)}|_{\mu^*} \ge PE_{b}^{k}|_{\mu^*}$, which is a contradiction.
 
 \textbf{Scenario 5:} $k \neq k'$, $\mu^*(b) \neq k'$, $\mathcal{B}_k^\mu \neq \mathcal{B}_k^{\mu^*}$, $\mathcal{B}_{\mu^*(b)}^\mu \neq \mathcal{B}_{\mu^*(b)}^{\mu^*}$ and the last reallocated PRB of UE $k$ was reallocated after the last reallocated PRB of UE $\mu^*(b)$ was reallocated.
 
 In this scenario, suppose the last reallocated PRB of UE $k$ was reallocated at time $t$ and the last reallocated PRB of UE $k$ is $b_k$. According to Alg. \ref{ma}, we know that $PE_b^{\mu^*(b)}(t) \ge PE_{b_k}^{k}(t)$. Then similar to the discussion in Scenario 4, we know that $PE_b^{\mu^*(b)}|_{\mu^*} \ge PE_{b}^{k}|_{\mu^*}$, which is a contradiction.
 
 \textbf{Scenario 6:} $k \neq k'$, $\mu^*(b) \neq k'$, $\mathcal{B}_k^\mu \neq \mathcal{B}_k^{\mu^*}$, $\mathcal{B}_{\mu^*(b)}^\mu \neq \mathcal{B}_{\mu^*(b)}^{\mu^*}$ and the last reallocated PRB of UE $k$ was reallocated before the last reallocated PRB of UE $\mu^*(b)$ was reallocated.
 
 In this scenario, suppose the last reallocated PRB of UE $k$ was reallocated at time $t$ and the last reallocated PRB of UE $k$ is $b_k$. According to Alg. \ref{ma}, we know that $PE_b^{\mu^*(b)}(t) \ge PE_{b_k}^{k}(t)$. According to the property of sigmoid function and Definition 2, we know that $PE_b^{\mu^*(b)}|_{\mu^*} \ge PE_b^{\mu^*(b)}(t)$ since that at least one PRB of UE $\mu^*(b)$ was reallocated after time $t$. And $PE_{b_k}^{k}|_{\mu^*} = PE_{b_k}^{k}(t)$ because no more PRB of UE $k$ was reallocated after time $t$. Hence, $PE_b^{\mu^*(b)}|_{\mu^*} \ge PE_{b_k}^{k}|_{\mu^*} = PE_{b}^{k}|_{\mu^*}$, which is a contradiction.
 
 \textbf{Scenario 7:} $\mu^*(b) = k'$.
 
 Assume that the last UE whose PRB was reallocated to UE $k'$ was UE $k^*$ and the last reallocated PRB was PRB $b^*$. Suppose PRB $b^*$ was reallocated at time $t$. Then there exist two subscenarios.
 
 \textbf{Subscenario 1:} $k = k^*$.
  Then according to Alg. \ref{ma}, $PE_{b^*}^{k'}(t) > PE_{b^*}^{k}(t)$. According to Definition 2, $PE_{b^*}^{k'}|_{\mu^*} = PE_{b^*}^{k'}(t)$ and $PE_{b^*}^{k}|_{\mu^*} = PE_{b^*}^{k}(t)$. Moreover, we know $PE_{b}^{k'}|_{\mu^*} = PE_{b^*}^{k'}|_{\mu^*}$ and $PE_{b}^{k}|_{\mu^*} = PE_{b^*}^{k}|_{\mu^*}$. Hence, $PE_{b}^{k'}|_{\mu^*} > PE_{b}^{k}|_{\mu^*}$. Because $\mu^*(b) = k'$, $PE_{b}^{\mu^*(b)}|_{\mu^*} > PE_{b}^{k}|_{\mu^*}$, which is a contradiction.
 
 \textbf{Subscenario 2:} $k \neq k^*$.
 This subscenario could be further divided into two subcases.
 
 \textbf{Subcase 1:} $\mathcal{B}_k^\mu = \mathcal{B}_k^{\mu^*}$.
 In this subcase, we know no PRB of UE $k$ was reallocated in the matching process. Since $\mu$ is a stable matching, we know $PE_{b^*}^{k^*}|_\mu \ge PE_{b^*}^k|_\mu$. Then according to Alg. \ref{ma}, $PE_{b^*}^{k'}(t) > PE_{b^*}^{k^*}(t)$. According to Definition 2, $PE_{b^*}^{k'}|_{\mu^*} = PE_{b^*}^{k'}(t)$ and $PE_{b^*}^{k^*}|_{\mu^*} = PE_{b^*}^{k}(t)$. And we know $PE_{b^*}^k|_\mu = PE_{b^*}^k|_{\mu^*}$ because $\mathcal{B}_k^\mu = \mathcal{B}_k^{\mu^*}$. Hence, we have $PE_{b^*}^{k'}|_{\mu^*} > PE_{b^*}^{k}|_{\mu^*}$. Moreover, because $PE_{b}^{k'}|_{\mu^*} = PE_{b^*}^{k'}|_{\mu^*}$ and $PE_{b}^{k}|_{\mu^*} = PE_{b^*}^{k}|_{\mu^*}$, we have $ PE_{b}^{\mu^*(b)}|_{\mu^*} = PE_{b}^{k'}|_{\mu^*} > PE_{b}^{k}|_{\mu^*}$, which is a contradiction.
 
 \textbf{Subcase 2:} $\mathcal{B}_k^\mu \neq \mathcal{B}_k^{\mu^*}$.
 In this subcase, we know that at least one PRB of UE $k$ was reallocated in the matching process. Suppose that the last reallocated PRB of UE $k$ was reallocated at time $t'$ and the last reallocated PRB of UE $k$ is $b'$. Then we know $t' < t$. According to Alg. \ref{ma}, we know that at time $t'$, $b'$ has the lowest preference extent on its matched UE, $k$. Hence, we have $PE_{b'}^k(t')\le PE_{b^*}^{k^*}(t')$. And according to Definition 2, $PE_{b}^k|_{\mu^*} = PE_{b}^k(t') = PE_{b'}^k(t')$ since no more PRB of UE $k$ has been reallocated after time $t'$. Moreover, we know $PE_{b^*}^{k^*}(t') \le PE_{b^*}^{k^*}(t)$ since no PRB was reallocated to UE $k^*$ between $t'$ and $t$. And since PRB $b^*$ was reallocated to UE $k'$ at time $t$, we know $PE_{b^*}^{k^*}(t) < PE_{b^*}^{k'}(t)$. According to Definition 2, $PE_{b^*}^{k'}|_{\mu^*} = PE_{b^*}^{k'}(t)$ and $PE_{b^*}^{k^*}|_{\mu^*} = PE_{b^*}^{k^*}(t)$. Hence, we have $ PE_{b}^k|_{\mu^*} < PE_{b}^{k'}|_{\mu^*} = PE_{b}^{\mu^*(b)}|_{\mu^*}$, which is a contradiction.
 
 From the discussion, we can conclude that there exist no blocking pair in matching $\mu^*$. Hence, matching $\mu^*$ is stable.
 \end{proof}

\begin{property}
The time complexity of Alg. \ref{ma} is $\mathcal{O}(K_j|B_{k'}|)$.
\end{property}

\begin{proof}
From Alg. \ref{ma}, we know that $|B_{k'}|$ iterations are needed for the while loop. And in each iteration, the most time-consuming part is the  first for loop, which requires $K_j$ iterations, each requires $3$ operations of $\mathcal{O}(1)$. Hence, the time complexity of Alg. \ref{ma} is $\mathcal{O}(K_j|B_{k'}|)$.
\end{proof}

\section{Interactive Optimization for Joint User Association and Resource Allocation}
The proposed interactive optimization algorithm for joint user association and resource allocation is shown in Alg. \ref{IO}. From Alg. \ref{IO}, we can find that it contains three stages, which are initialization, correction and optimization. The purpose of introduction of correction stage is to reassociate the UEs which are not allocated any PRB by their associated PBS with their nearest MBS, since that these PBSs don't have enough PRBs to support their demands.

In the optimization stage, user association and resource allocation are solved iteratively in an interactive way until convergence is achieved. Moreover, user association in the optimization stage is performed in a heuristic way in order to avoid exhaustive search, whose time complexity is too high. The first step is to calculate the set of UEs which could be reassociated, which are in fact the UEs associated with PBSs. Moreover, the indicator matrix, $\Psi$, which is an $M$-by-$N$ matrix, is initialized. The indicator matrix $\Psi$ will be used later for guiding the selection of UE which will be reassociated. 

The following part is a while loop which ends when there exists no UE which could be reassociated. Each iteration of the while loop contains a for loop, which loops over all the set of MBSs. For each MBS, the set of UEs covered by it, $\mathcal{K'}_j$, is calculated. Then UE $k \in \mathcal{K'}_j$ is searched which has the lowest utility, and its corresponding PBS $i$ is found. It is probably that the burden of PBS $i$ is very high. Hence, the UEs in the set $\mathcal{K'}_j \cap \mathcal{K}_i$ could be considered for reassociation. If the indicator $\psi_{j,i}$ equals 0, it means that the UE with lowest utility in $\mathcal{K'}_j \cap \mathcal{K}_i$ could be considered for reassociation. Then UE $k$ will be reassociated to see whether this brings any gain of the sum utilities of all the UEs. If not, the former UE association policy will be restored. If $\psi_{j,i} = 1$, it means the UE with lowest utility in $\mathcal{K'}_j \cap \mathcal{K}_i$ has been tried for reassociation but this failed to bring any sum utility gain. Hence, on the contrary, the UE with highest utility in $\mathcal{K'}_j \cap \mathcal{K}_i$ will be tried for reassociation. If this also doesn't bring any sum utility gain of all the UEs, all the UEs in the set $\mathcal{K'}_j \cap \mathcal{K}_i$ will not be considered for reassociation anymore.

\begin{algorithm*}[hbtp]
\caption{Interactive Optimization Algorithm for Joint User Association and Resource Allocation}
\label{IO}
\begin{multicols}{2}
\begin{algorithmic}
\STATE \textbf{1. Initialization:}\\
\STATE Invoke Alg. \ref{UAInitialization} to perform user association and get the user association matrix $X$;\\
\FOR{$j \in \mathcal{F}$}
\STATE Invoke Alg. \ref{ODA} to allocate PRBs to UEs associated with BS $j$;
\STATE Invoke Alg. \ref{PRA} to allocate power to PRBs of BS $j$;
\ENDFOR
\STATE \textbf{2. Correction:}\\
\FOR{$k \in \mathcal{K}$}
\STATE Calculate the set of PRBs allocated to UE $k$, $\mathcal{B}_k$;\\
\IF{$\mathcal{B}_k = \emptyset$ and UE $k$ is associated with a pico base station}
\STATE Reassociate UE $k$ with the MBS which is nearest to it;
\ENDIF
\ENDFOR
\FOR{$j \in \mathcal{F}$}
\STATE Invoke Alg. \ref{ODA} to allocate PRBs to UEs associated with BS $j$;
\STATE Invoke Alg. \ref{PRA} to allocate power to PRBs of BS $j$;
\ENDFOR
\STATE \textbf{3. Optimization:}\\
\STATE Calculate the sum of UEs' utilities, $\sum_{k \in \mathcal{K}}U_k$;\\
\STATE Calculate the set of UEs which could be reassociated, $\mathcal{K'}$;\\
\STATE Initialize indicator matrix $\Psi = O_{M \times N}$;\\
\WHILE{$\mathcal{K'} \neq \emptyset$}
\FOR{$j \in \mathcal{F}_{MBS}$}
\STATE Calculate the set of UEs in $\mathcal{K'}$ located within the coverage area of MBS $j$, $\mathcal{K'}_j$;
\IF{$\mathcal{K'}_j = \emptyset$}
\STATE \textbf{Continue};\\
\ENDIF
\STATE Find the UE $k$ in $\mathcal{K'}_j$ with the lowest utility;
\STATE Find the PBS $i$ which UE $k$ is associated with;
\IF{$\psi_{j,i} = 0$}
\STATE $k' = k$;\\
\ELSE
\STATE Find the UE $k' \in \mathcal{K'}_j \cap \mathcal{K}_i$  with the highest utility;\\
\ENDIF
\STATE Reassociate UE $k'$ with MBS $j$;
\STATE Invoke Alg. \ref{md} to generate a new PRB-UE matching for PBS $i$;
\STATE Invoke Alg. \ref{PRA} to allocate power to PRBs of PBS $i$;
\STATE Invoke Alg. \ref{ma} to generate a new PRB-UE matching for MBS $j$;
\STATE Invoke Alg. \ref{PRA} to allocate power to PRBs of MBS $j$;
\STATE Recalculate the sum of UEs' utilities, $\sum_{k \in \mathcal{K}}U'_k$;

\IF{$\sum_{k \in \mathcal{K}}U'_k < \sum_{k \in \mathcal{K}}U_k$}
\STATE Reassociate UE $k'$ with PBS $i$;
\STATE Recover the former PRB-UE matching and power allocation results of MBS $j$ and PBS $i$;\\
%\STATE $\mathcal{K'} = \mathcal{K'} - \mathcal{K'}_{j} \cap \mathcal{K}_i$;
\IF{$\psi_{j,i} = 0$}
\STATE $\psi_{j,i} = 1$;\\
\ELSE
\STATE $\mathcal{K'}_j = \mathcal{K'}_j - \mathcal{K'}_{j} \cap \mathcal{K}_i$;
\ENDIF
\ELSE
\STATE $\mathcal{K'}_j = \mathcal{K'}_j - \{k'\}$;\\
\ENDIF
\ENDFOR
\FOR{$j \in \mathcal{F}_{MBS}$}
\STATE Invoke Alg. \ref{PRA} to allocate power to PRBs of BS $j$;\\
\ENDFOR
\ENDWHILE
\end{algorithmic}
\end{multicols}
\end{algorithm*}

\section{Performance Evaluation}
\subsection{Simulation Settings}
We consider a horizontal area of $2000$ meter $\times$ $ 2000$ meter. The main parameter settings are summarized in Table \ref{ParameterSettings}. We define an area as a sparse area if it is covered by no PBS. And we define an area as a dense area if it is covered by a PBS. For simplicity, we suppose that there exist two types of UEs, which are enhanced mobile broadband (eMBB) UEs and ultra-reliable and low-latency communications (uRLLC) UEs. The UEs within the same type have the same QoS requirements. However, different UEs within the same type could have different utility functions. For simulation, we suppose that if UE $k$ is an eMBB UE, $w_k^1$ is uniformly chosen from 0.8 to 0.9 and $w_k^2 = 1 - w_k^1$. And we suppose that if UE $k$ is a uRLLC UE, $w_k^2$ is uniformly chosen from 0.8 to 0.9 and $w_k^1 = 1 - w_k^2$.

To better evaluate the performance of the proposed IOA algorithm, we choose different settings of number of PBSs and maximum amount of power per PBS. Specifically, we choose three different settings of number of PBSs, which are 9, 18 and 27. In addition, we choose ten different settings of maximum amount of transmit power per PBS, from 0.1 W to 1.0 W.

\begin{table}

\centering

\caption{Summary of the Main Parameter Settings.}

\label{ParameterSettings}

\begin{tabular}{c|c}

  \hline
  Parameter & Value\\
  \hline
  Coverage radius of MBS & 500 m\\
  \hline
  Coverage radius of PBS & 100 m\\
  \hline
  Number of MBSs & 9\\
  \hline
  Density of eMBB UEs in a sparse area & 8 / $\textrm{km}^\textrm{2}$\\
  \hline
  Density of uRLLC UEs in a sparse area & 8 / $\textrm{km}^\textrm{2}$\\
  \hline
  Density of eMBB UEs in a dense area & 100 / $\textrm{km}^\textrm{2}$\\
  \hline
  Density of uRLLC UEs in a dense area & 100 / $\textrm{km}^\textrm{2}$\\
  \hline
  Data transmission rate requirement & \multirow{2}{*}{100 Mbps}\\
  of eMBB UEs & \\
  \hline
  Latency requirement of eMBB UEs & 50 ms\\
  \hline
  Maximum BER requirement of eMBB UEs & $\textrm{1} \times \textrm{10}^{-4}$\\
  \hline
  Maximum BER requirement of uRLLC UEs & $\textrm{1} \times \textrm{10}^{-6}$\\
  \hline
  Data transmission rate requirement & \multirow{2}{*}{1 Mbps}\\
  of uRLLC UEs & \\
  \hline
  Latency requirement of uRLLC UEs & 20 ms\\
  \hline
  Total bandwidth per base station & 100 MHz\\
  \hline
  Subcarrier spacing & 30 KHz\\
  \hline
  Total number of PRBs per base station & 273\\
  \hline
  Frequency reuse factor & 1/3\\
  \hline
  Maximum transmit power per MBS & 40 W\\
  \hline
  Carrier frequency of MBS & 3.5 GHz\\
  \hline
  Carrier frequency of PBS & 60 GHz\\
  \hline
  \multirow{2}{*}{Distance-dependent path loss of MBS (dB)} & 36$\textrm{log}_{\textrm{10}}$\textrm{(d) + 29.358,}\\
   & d in m\\
  \hline
  \multirow{2}{*}{Distance-dependent path loss of PBS (dB)} & 44$\textrm{log}_{\textrm{10}}$\textrm{(d) + 43.985,}\\
   & d in m\\
  \hline
  Small-scale fading model & Rayleigh\\
  \hline
  AWGN power & -174 dBm/Hz\\
  \hline
  Packet size & 1000 bit\\
  \hline
  Average latency from server to base station & \multirow{2}{*}{30 ms}\\
  of eMBB UEs & \\
  \hline
  Average latency from server to base station & \multirow{2}{*}{15 ms}\\
  of uRLLC UEs & \\
  \hline
  Packet arrival rate of eMBB UEs & 80000 / s\\
  \hline
  Packet arrival rate of uRLLC UEs & 800 / s\\
  \hline
  Propagation latency & 1 $\mu\textrm{s}$\\
  \hline
  
\end{tabular}

\end{table}

The distributions of base stations and UEs under the three settings of number of PBSs is illustrated in Fig. \ref{distirbution}.

\begin{figure*}[htbp]
  \begin{center}
  \includegraphics[width=7.0in]{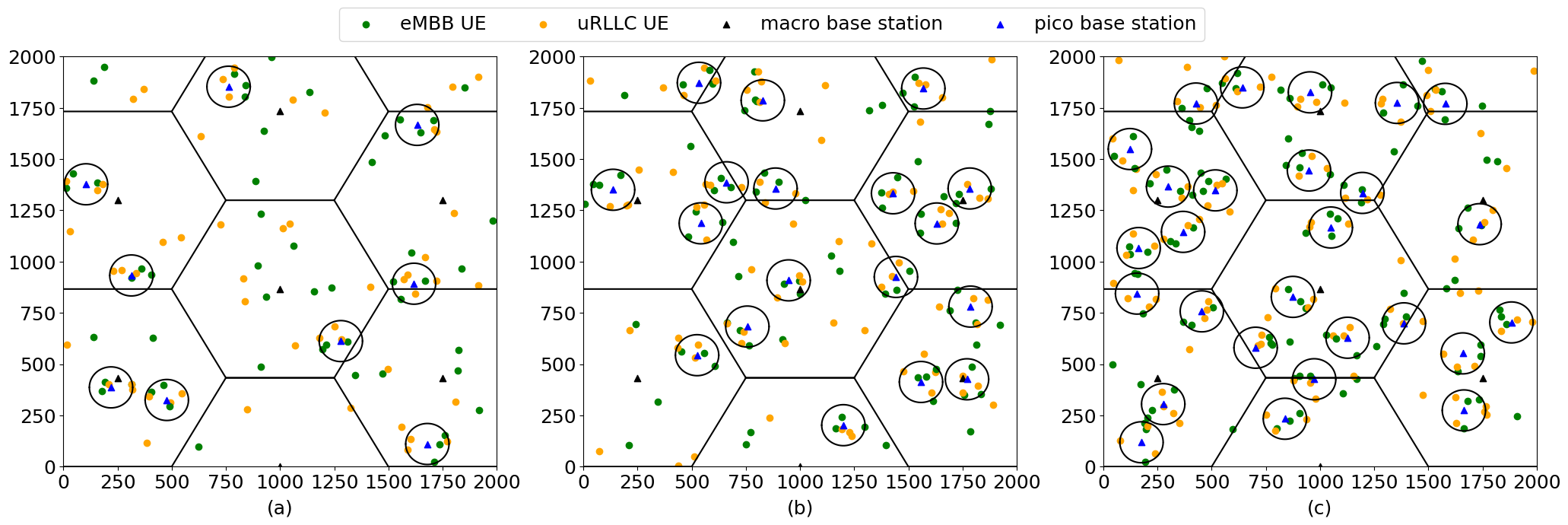}\\
  \caption{Illustration of distributions of base stations and UEs under (a) number of PBSs = 9, (b) number of PBSs = 18, (c) number of PBSs = 27. The MBSs are distributed in a planned cellular manner and the PBSs are distributed randomly. The UEs are classified into two groups, which are eMBB and uRLLC UEs, and are distributed randomly.}\label{distirbution}
  \end{center}
\end{figure*}

\subsection{Convergence Behavior of IOA Algorithm}
The convergence behavior of IOA algorithm under different settings is illustrated in Fig. \ref{convergence}. From Fig. \ref{convergence}, we can find that under all settings of number of PBSs and maximum transmit power per PBS, the proposed IOA algorithm could achieve convergence within a limited number of iterations.

\begin{figure*}[htbp]
  \begin{center}
  \includegraphics[width=7.0in]{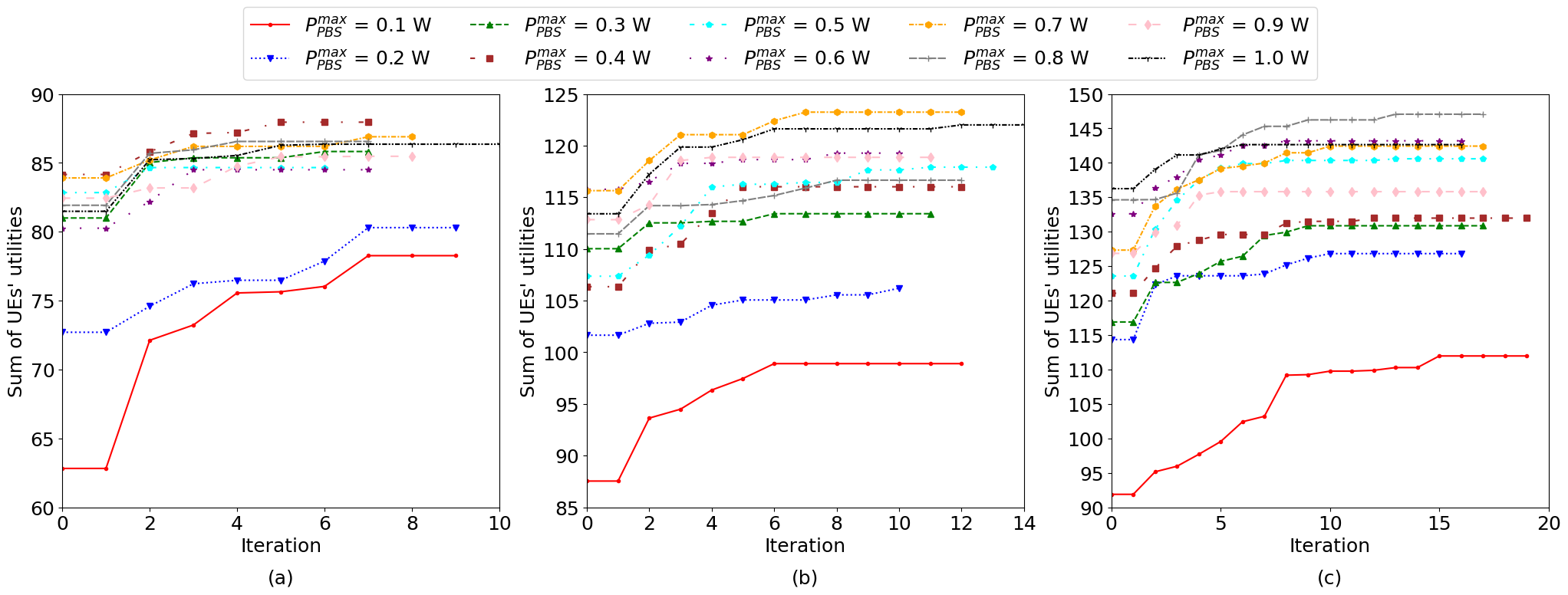}\\
  \caption{Illustration of convergence of IOA algorithm under (a) number of PBSs = 9, (b) number of PBSs = 18, (c) number of PBSs = 27 with ten different settings of maximum transmit power per PBS, ranging from 0.1 W  to 1.0 W. Under each setting, convergence could be achieved by IOA algorithm within a limited number of iterations.}\label{convergence}
  \end{center}
\end{figure*}

\subsection{Performance Comparison}
The whole user association and resource allocation task could be divided into three subtasks, which are user association, PRB allocation and power allocation, respectively. For user association, we select three comparison algorithms, which are random, max-Reference signal received power (RSRP) \cite{liu2017backhaul} and max-biased RSRP \cite{8600150}  user association algorithms. For PRB allocation, we select four comparison algorithms, which are uniform, round robin \cite{barayan2013performance}, maximum sum rate \cite{zhang2004multiuser} and max-min fair \cite{cai2008comparision} PRB allocation algorithms. For power allocation, we select two comparison algorithms, which are uniform  and water-filling \cite{ling2013fast} power allocation algorithms.

By combining the comparison algorithms of the three tasks, seven baseline algorithms are created, as summarized in Table \ref{ComparisonAlgorithms}. Specifically, for max-biased RSRP user association algorithm, we set the coefficient which is multiplied by RSRP of PBSs 
to be 100. 

\begin{table*}

\centering

\caption{Summary of the Baseline Algorithms.}

\label{ComparisonAlgorithms}

\begin{tabular}{c|c}

  \hline
  Baseline algorithm & Composition\\
  \hline
  BA1 & Random user association + uniform PRB allocation + uniform power allocation\\
  \hline
  BA2 & Max-RSRP user association + round robin PRB allocation + water-filling power allocation\\
  \hline
  BA3 & Max-RSRP user association + maximum sum rate PRB allocation + water-filling power allocation\\
  \hline
  BA4 & Max-RSRP user association + max-min fair PRB allocation + water-filling power allocation\\
  \hline
  BA5 & Max-biased RSRP user association + round robin PRB allocation + water-filling power allocation\\
  \hline
  BA6 & Max-biased RSRP user association + maximum sum rate PRB allocation + water-filling power allocation\\
  \hline
  BA7 & Max-biased RSRP user association + max-min fair PRB allocation + water-filling power allocation\\
  \hline
  
\end{tabular}

\end{table*}

In order to better conduct performance evaluation, we select two performance metrics, which are average utility and UE satisfaction ratio. Average utility is defined as the average utility of UEs in the network. UE satisfaction ratio is defined as the ratio of UEs whose QoS requirements are satisfied. 

The comparison of average utility under different settings is shown in Fig. \ref{averageUtility}. From Fig. \ref{averageUtility}, we can find that the proposed IOA algorithm achieves the highest average utility under all the settings of number of PBSs and maximum transmit power per PBS. We can also find that BA1 achieves the second best results when number of PBSs is equal to 9 or 18. And BA1 achieves nearly the highest average utilities among all the baseline algorithms when number of PBSs is equal to 27. This is because the random user association adopted in BA1 enables the PBSs to serve more UEs to achieve load balancing. 

Moreover, we can find that BA5 outperforms BA2, BA6 outperforms BA3 and BA7 outperforms BA4. The reason is that max-biased RSRP user association allows more UEs to be served by PBSs compared with max-RSRP user association, which almost associates all the UEs with MBSs. In addition, we can find that BA3 and BA6 achieves the worst results. This is because the maximum sum rate PRB allocation algorithm allocates each PRB to the UE with the best channel condition to maximize the sum rate, which sacrifices the UEs with worse channel conditions. On the contrary, round robin and max-min fair PRB allocation could better ensure fairness among UEs with different channel conditions.

\begin{figure*}[htbp]
  \begin{center}
  \includegraphics[width=7.0in]{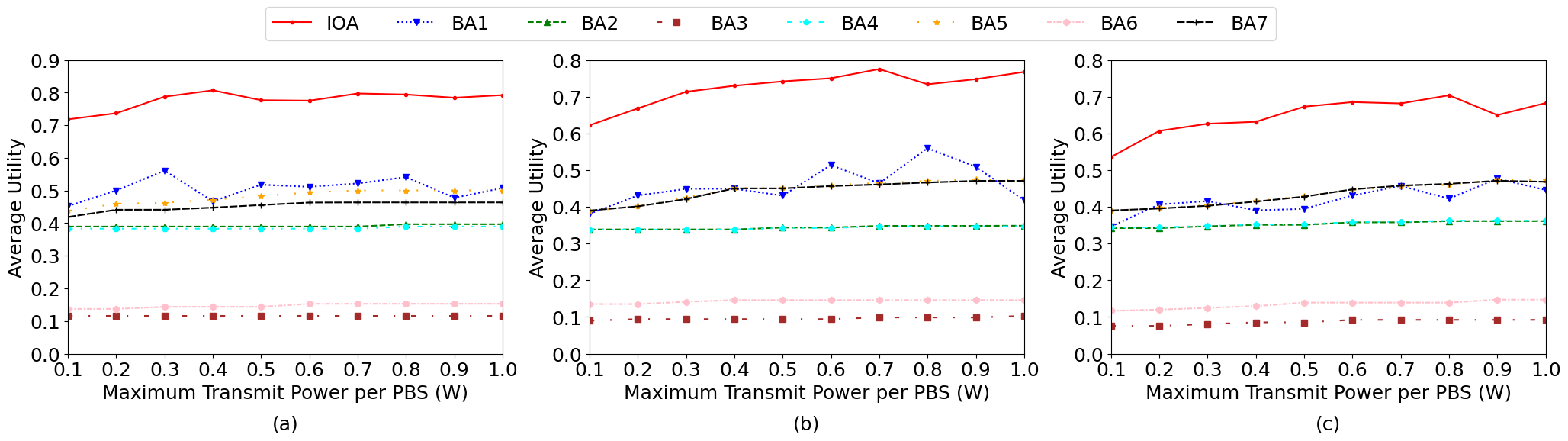}\\
  \caption{Comparison of average utility among IOA and seven baseline algorithms under (a) number of PBSs = 9, (b) number of PBSs = 18, (c) number of PBSs = 27 with ten different settings of maximum transmit power per PBS ranging from 0.1 W to 1.0 W. Under each setting, IOA algorithm achieves the highest average utility.}\label{averageUtility}
  \end{center}
\end{figure*}

The comparison of UE satisfaction ratio under different settings is shown in Fig. \ref{UEsatisfactionRatio}. From Fig. \ref{UEsatisfactionRatio}, we can find that the proposed IOA algorithm achieves the highest UE satisfaction ratio under all the settings of number of PBSs and maximum transmit power per PBS. Comparing Fig. \ref{UEsatisfactionRatio} with Fig. \ref{averageUtility}, we can find that there exists relevance between average utility and UE satisfaction ratio, which means that a higher average utility usually indicates a higher UE satisfaction ratio.

\begin{figure*}[htbp]
  \begin{center}
  \includegraphics[width=7.0in]{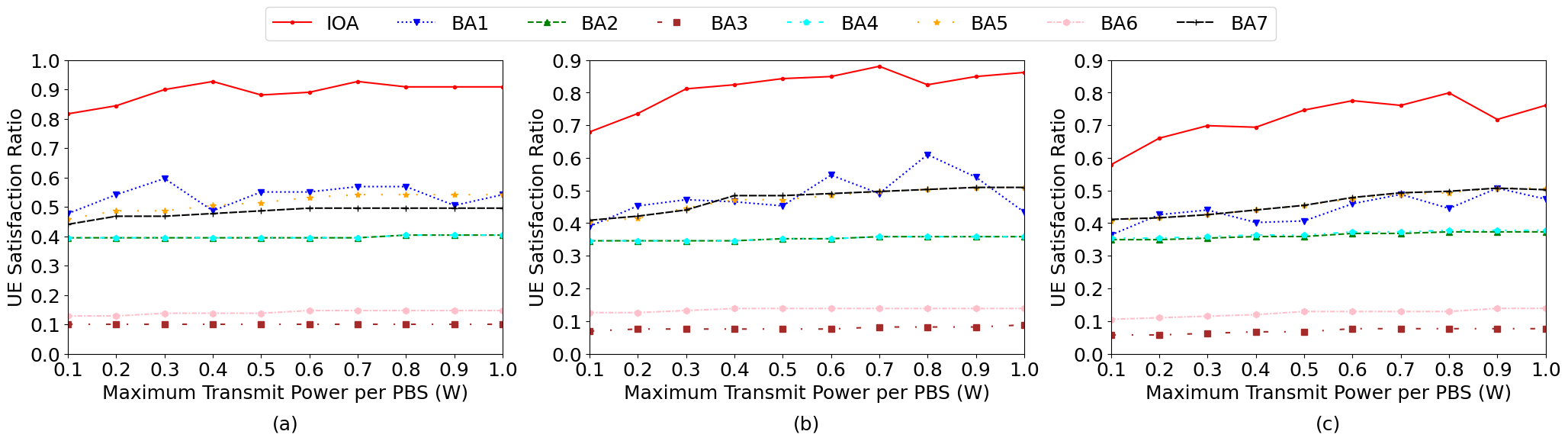}\\
  \caption{Comparison of UE satisfaction ratio among IOA and seven baseline algorithms under (a) number of PBSs = 9, (b) number of PBSs = 18, (c) number of PBSs = 27 with ten different settings of maximum transmit power per PBS ranging from 0.1 W to 1.0 W. Under each setting, IOA algorithm achieves the highest UE satisfaction ratio.}\label{UEsatisfactionRatio}
  \end{center}
\end{figure*}

The simulation results confirm that our proposed IOA algorithm achieves much better performance compared with the baseline algorithms.

\section{Conclusion}
 In this paper, the problem of application-specific objective-based joint user association and resource allocation for tailored QoS provisioning in downlink 6G HetNets was investigated. We formulated this problem as an MINLP problem, aiming at maximizing the sum of UEs' application-specific objectives. For solving this non-convex and NP-hard problem, an interactive optimization algorithm was proposed which solves the problems of user association and resource allocation iteratively in an interactive way until convergence is achieved. Extensive experiments were conducted, whose results confirm that our proposed IOA algorithm outperforms baseline algorithms in terms of both average utility and UE satisfaction ratio.

\ifCLASSOPTIONcaptionsoff
  \newpage
\fi

\bibliographystyle{IEEEtran}
\bibliography{IEEEabrv,Bibliography}

\begin{thebibliography}{10}
\providecommand{\url}[1]{#1}
\csname url@rmstyle\endcsname
\providecommand{\newblock}{\relax}
\providecommand{\bibinfo}[2]{#2}
\providecommand\BIBentrySTDinterwordspacing{\spaceskip=0pt\relax}
\providecommand\BIBentryALTinterwordstretchfactor{4}
\providecommand\BIBentryALTinterwordspacing{\spaceskip=\fontdimen2\font plus
\BIBentryALTinterwordstretchfactor\fontdimen3\font minus \fontdimen4\font\relax}
\providecommand\BIBforeignlanguage[2]{{%
\expandafter\ifx\csname l@#1\endcsname\relax
\typeout{** WARNING: IEEEtran.bst: No hyphenation pattern has been}%
\typeout{** loaded for the language `#1'. Using the pattern for}%
\typeout{** the default language instead.}%
\else
\language=\csname l@#1\endcsname
\fi
#2}}

\bibitem{10413956}
K.~Qu, W.~Zhuang, Q.~Ye, W.~Wu, and X.~Shen, ``Model-assisted learning for adaptive cooperative perception of connected autonomous vehicles,'' \emph{IEEE Trans. Wirel. Commun.}, 2024, (Early Access).

\bibitem{8933555}
R.~Gupta, S.~Tanwar, S.~Tyagi, and N.~Kumar, ``Tactile-internet-based telesurgery system for healthcare 4.0: An architecture, research challenges, and future directions,'' \emph{IEEE Netw.}, vol.~33, no.~6, pp. 22--29, 2019.

\bibitem{9068423}
I.-S. Comşa, G.-M. Muntean, and R.~Trestian, ``An innovative machine-learning-based scheduling solution for improving live {UHD} video streaming quality in highly dynamic network environments,'' \emph{IEEE Trans. Multimed.}, vol.~67, no.~1, pp. 212--224, 2021.

\bibitem{8713498}
T.~Dang and M.~Peng, ``Joint radio communication, caching, and computing design for mobile virtual reality delivery in fog radio access networks,'' \emph{IEEE J. Sel. Areas Commun.}, vol.~37, no.~7, pp. 1594--1607, 2019.

\bibitem{7994919}
K.~Chi, L.~Huang, Y.~Li, Y.-h. Zhu, X.-z. Tian, and M.~Xia, ``Efficient and reliable multicast using device-to-device communication and network coding for a {5G} network,'' \emph{IEEE Netw.}, vol.~31, no.~4, pp. 78--84, 2017.

\bibitem{6692187}
A.~Stavridis, S.~Sinanovic, M.~Di~Renzo, and H.~Haas, ``Energy evaluation of spatial modulation at a multi-antenna base station,'' in \emph{Proc. IEEE VTC Fall}, Las Vegas, NV, USA, 2013, pp. 1--5.

\bibitem{6133859}
H.~Mehrpouyan, S.~D. Blostein, and T.~Svensson, ``A new distributed approach for achieving clock synchronization in heterogeneous networks,'' in \emph{Proc. IEEE GLOBECOM}, Houston, TX, USA, 2011, pp. 1--5.

\bibitem{10234545}
P.~Jia and X.~Wang, ``A new virtual network topology-based digital twin for spatial-temporal load-balanced user association in {6G} {HetNets},'' \emph{IEEE J. Sel. Areas Commun.}, vol.~41, no.~10, pp. 3080--3094, 2023.

\bibitem{10419174}
A.~Alizadeh, B.~Lim, and M.~Vu, ``Multi-agent {Q}-learning for real-time load balancing user association and handover in mobile networks,'' \emph{IEEE Trans. Wirel. Commun.}, 2024, (Early Access).

\bibitem{9973061}
W.~Wu, F.~Yang, F.~Zhou, Q.~Wu, and R.~Q. Hu, ``Intelligent resource allocation for {IRS}-enhanced {OFDM} communication systems: A hybrid deep reinforcement learning approach,'' \emph{IEEE Trans. Wirel. Commun.}, vol.~22, no.~6, pp. 4028--4042, 2023.

\bibitem{wang2017joint}
F.~Wang, W.~Chen, H.~Tang, and Q.~Wu, ``Joint optimization of user association, subchannel allocation, and power allocation in multi-cell multi-association {OFDMA} heterogeneous networks,'' \emph{IEEE Trans. Commun.}, vol.~65, no.~6, pp. 2672--2684, 2017.

\bibitem{chen2015joint}
Y.~Chen, J.~Li, W.~Chen, Z.~Lin, and B.~Vucetic, ``Joint user association and resource allocation in the downlink of heterogeneous networks,'' \emph{IEEE Trans. Veh. Technol.}, vol.~65, no.~7, pp. 5701--5706, 2015.

\bibitem{10211332}
B.~Agarwal, M.~A. Togou, M.~Ruffini, and G.-M. Muntean, ``A low complexity {ML}-assisted multi-knapsack-based approach for user association and resource allocation in {5G} {HetNets},'' in \emph{Proc. IEEE BMSB}, Beijing, China, 2023, pp. 1--6.

\bibitem{cheng2020joint}
Z.~Cheng, N.~Chen, B.~Liu, Z.~Gao, L.~Huang, X.~Du, and M.~Guizani, ``Joint user association and resource allocation in {HetNets} based on user mobility prediction,'' \emph{Comput. Netw.}, vol. 177, p. 107312, 2020.

\bibitem{10184114}
W.~Deng, Y.~Liu, M.~Li, and M.~Lei, ``{GNN}-aided user association and beam selection for {mmWave}-integrated heterogeneous networks,'' \emph{IEEE Wirel. Commun. Lett.}, vol.~12, no.~11, pp. 1836--1840, 2023.

\bibitem{han2016backhaul}
Q.~Han, B.~Yang, G.~Miao, C.~Chen, X.~Wang, and X.~Guan, ``Backhaul-aware user association and resource allocation for energy-constrained {HetNets},'' \emph{IEEE Trans. Veh. Technol.}, vol.~66, no.~1, pp. 580--593, 2016.

\bibitem{vaezpour2019robust}
E.~Vaezpour, M.~Dehghan, and H.~Yousefi’zadeh, ``Robust joint user association and resource partitioning in heterogeneous cloud {RANs} with dual connectivity,'' \emph{Comput. Commun.}, vol. 138, pp. 1--10, 2019.

\bibitem{liu2019joint}
R.~Liu, Q.~Chen, G.~Yu, and G.~Y. Li, ``Joint user association and resource allocation for multi-band millimeter-wave heterogeneous networks,'' \emph{IEEE Trans. Commun.}, vol.~67, no.~12, pp. 8502--8516, 2019.

\bibitem{somesula2022artificial}
S.~Somesula, N.~Sharma, and A.~Anpalagan, ``Artificial bee optimization aided joint user association and resource allocation in hcran,'' \emph{Appl. Soft Comput.}, vol. 125, p. 109152, 2022.

\bibitem{liu2020joint}
J.-S. Liu, C.-H.~R. Lin, and Y.-C. Hu, ``Joint resource allocation, user association, and power control for {5G} {LTE}-based heterogeneous networks,'' \emph{IEEE Access}, vol.~8, pp. 122\,654--122\,672, 2020.

\bibitem{MUGHEES2023102206}
A.~Mughees, M.~Tahir, M.~A. Sheikh, A.~Amphawan, Y.~K. Meng, A.~Ahad, and K.~Chamran, ``Energy-efficient joint resource allocation in {5G} {HetNet} using multi-agent parameterized deep reinforcement learning,'' \emph{Phys. Commun.}, vol.~61, p. 102206, 2023.

\bibitem{10246069}
A.~Mughees, M.~Tahir, M.~A. Sheikh, A.~Amphawan, K.~M. Yap, M.~H. Habaebi, and M.~R. Islam, ``Energy efficient joint user association and power allocation using parameterized deep {DQN},'' in \emph{Proc. IEEE ICCCE}, Kuala Lumpur, Malaysia, 2023, pp. 35--40.

\bibitem{li2016joint}
Y.~Li, M.~Sheng, Y.~Sun, and Y.~Shi, ``Joint optimization of {BS} operation, user association, subcarrier assignment, and power allocation for energy-efficient {HetNets},'' \emph{IEEE J. Sel. Areas Commun.}, vol.~34, no.~12, pp. 3339--3353, 2016.

\bibitem{yin2019energy}
F.~Yin, A.~Wang, D.~Liu, and Z.~Zhang, ``Energy-aware joint user association and resource allocation for coded cache-enabled {HetNets},'' \emph{IEEE Access}, vol.~7, pp. 94\,128--94\,142, 2019.

\bibitem{10254460}
Y.~Kim, J.~Jang, and H.~J. Yang, ``Distributed resource allocation and user association for max-min fairness in {HetNets},'' \emph{IEEE Trans. Veh. Technol.}, vol.~73, no.~2, pp. 2983--2988, 2024.

\bibitem{10264114}
B.~Huang and A.~Guo, ``A dynamic hierarchical game approach for user association and resource allocation in {HetNets} with wireless backhaul,'' \emph{IEEE Wirel. Commun. Lett.}, vol.~13, no.~1, pp. 59--63, 2024.

\bibitem{noorivatan2020joint}
N.~Noorivatan and B.~Mahboobi, ``Joint user association and power-bandwidth allocation in heterogeneous cellular networks,'' in \emph{Proc. IEEE IST}, Tehran, Iran, 2020, pp. 96--102.

\bibitem{sokun2017novel}
H.~U. Sokun, R.~H. Gohary, and H.~Yanikomeroglu, ``A novel approach for {QoS}-aware joint user association, resource block and discrete power allocation in {HetNets},'' \emph{IEEE Trans. Wirel. Commun.}, vol.~16, no.~11, pp. 7603--7618, 2017.

\bibitem{chaieb2020joint}
C.~Chaieb, F.~Abdelkefi, and W.~Ajib, ``Joint user association and sub-channel assignment in wireless networks with heterogeneous multiple access and heterogeneous base stations,'' in \emph{Proc. IEEE PIMRC}, London, UK, 2020, pp. 1--6.

\bibitem{zhao2019deep}
N.~Zhao, Y.-C. Liang, D.~Niyato, Y.~Pei, M.~Wu, and Y.~Jiang, ``Deep reinforcement learning for user association and resource allocation in heterogeneous cellular networks,'' \emph{IEEE Trans. Wirel. Commun.}, vol.~18, no.~11, pp. 5141--5152, 2019.

\bibitem{7266476}
M.~A. AboulHassan, M.~Yassin, S.~Lahoud, M.~Ibrahim, D.~Mezher, B.~Cousin, and E.~A. Sourour, ``Classification and comparative analysis of inter-cell interference coordination techniques in {LTE} networks,'' in \emph{Proc. IEEE NTMS}, Paris, France, 2015, pp. 1--6.

\bibitem{9735275}
B.~Shi, F.-C. Zheng, C.~She, J.~Luo, and A.~G. Burr, ``Risk-resistant resource allocation for {eMBB} and {URLLC} coexistence under {M/G/1} queueing model,'' \emph{IEEE Trans. Veh. Technol.}, vol.~71, no.~6, pp. 6279--6290, 2022.

\bibitem{liu2017backhaul}
D.~Liu, Y.~Chen, K.~K. Chai, and T.~Zhang, ``Backhaul aware joint uplink and downlink user association for delay-power trade-offs in {HetNets} with hybrid energy sources,'' \emph{Trans. Emerg. Telecommun. Technol.}, vol.~28, no.~3, p. e2968, 2017.

\bibitem{8600150}
F.~Shi, K.~Sun, W.~Huang, and Y.~Wei, ``User association for on-grid energy minimizing in {HetNets} with hybrid energy supplies,'' in \emph{Proc. IEEE ICCT}, Chongqing, China, 2018, pp. 778--783.

\bibitem{barayan2013performance}
Y.~Barayan and I.~Kostanic, ``Performance evaluation of proportional fairness scheduling in {LTE},'' in \emph{Proc. WCECS}, vol.~2, San Francisco, USA, 2013, pp. 712--717.

\bibitem{zhang2004multiuser}
Y.~J. Zhang and K.~B. Letaief, ``Multiuser adaptive subcarrier-and-bit allocation with adaptive cell selection for {OFDM} systems,'' \emph{IEEE Trans. Wirel. Commun.}, vol.~3, no.~5, pp. 1566--1575, 2004.

\bibitem{cai2008comparision}
Y.~Cai, J.~Yu, Y.~Xu, and M.~Cai, ``A comparision of packet scheduling algorithms for {OFDMA} systems,'' in \emph{Proc. IEEE ICSPCS}, Gold Coast, QLD, Australia, 2008, pp. 1--5.

\bibitem{ling2013fast}
X.~Ling, B.~Wu, H.~Wen, L.~Pan, and F.~Luo, ``Fast and efficient parallel-shift water-filling algorithm for power allocation in orthogonal frequency division multiplexing-based underlay cognitive radios,'' \emph{IET Commun.}, vol.~7, no.~12, pp. 1269--1278, 2013.

\end{thebibliography}

\begin{IEEEbiography}[{\includegraphics[width=1in,height=1.25in,clip,keepaspectratio]{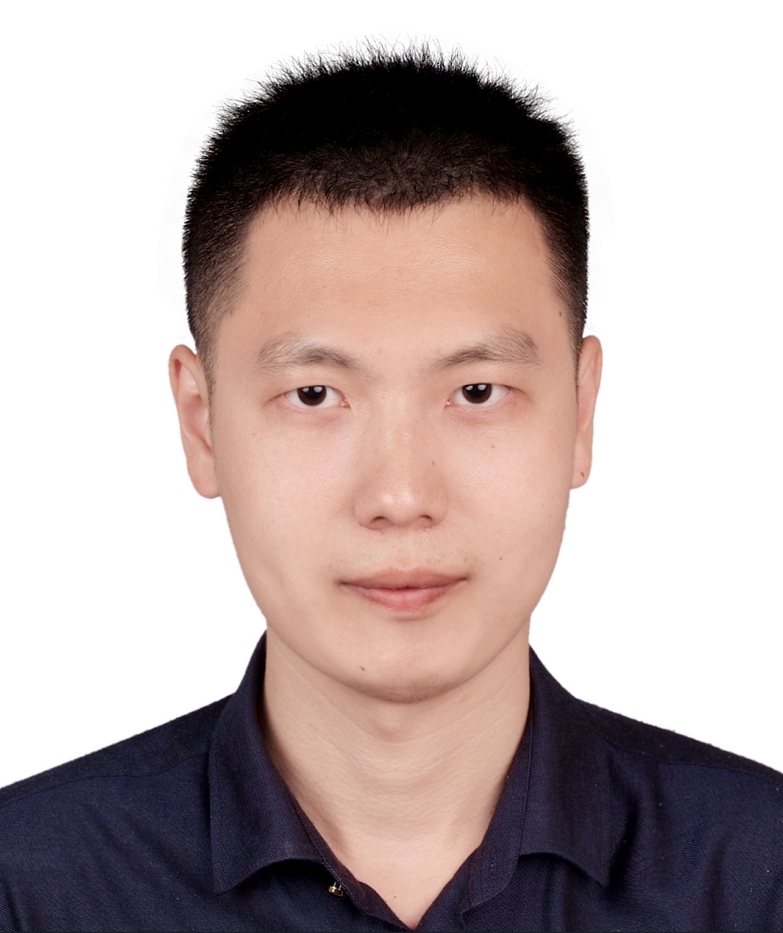}}]{Yongqin Fu}
    (Graduate Student Member, IEEE) received the B.E. degree in information engineering from Southeast University, Nanjing, China, in 2015 and the M.E. degree in computer science and technology from Zhejiang University, Hangzhou, China, in 2019. He is currently working towards the Ph.D. degree at the Department of Electrical and Computer Engineering, Western University, London, Ontario, Canada. 
    
    His current research interests include intelligent and customized resource management and network cooperation in beyond 5G and 6G networks, as well as sensor data analysis in Internet of Things (IoT) systems.
\end{IEEEbiography}

\begin{IEEEbiography}[{\includegraphics[width=1in,height=1.25in,clip,keepaspectratio]{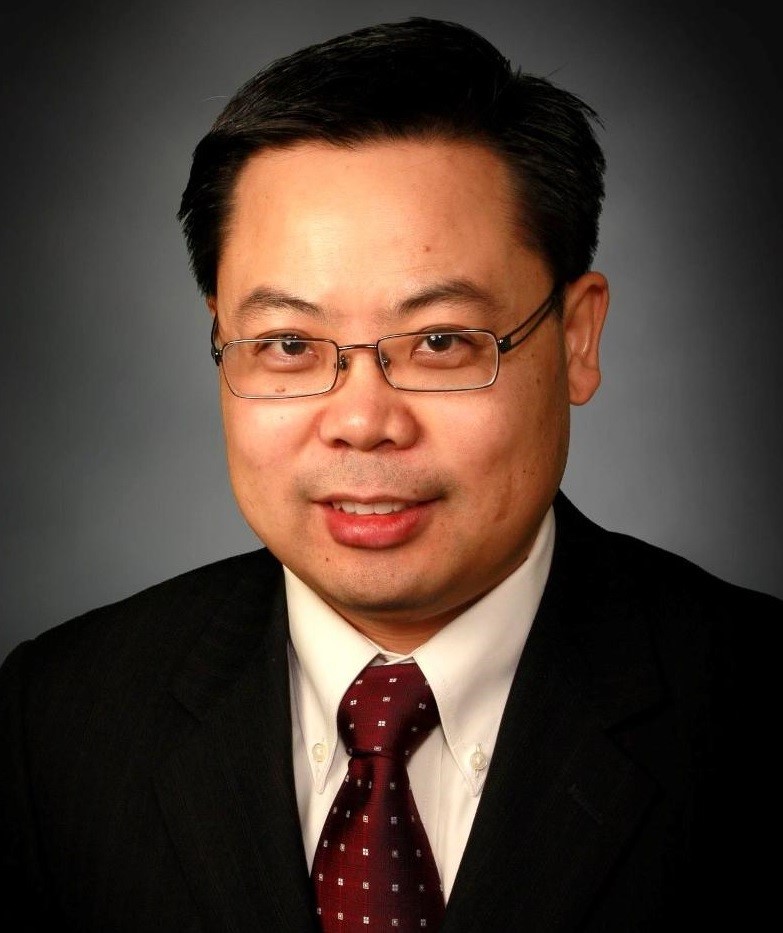}}]{Xianbin Wang}
   (Fellow, IEEE) received his Ph.D. degree in electrical and computer engineering from the National University of Singapore in 2001.
   
He is a Professor and a Tier-1 Canada Research Chair in 5G and Wireless IoT Communications with Western University, Canada. Prior to joining Western University, he was with the Communications Research Centre Canada as a Research Scientist/Senior Research Scientist from 2002 to 2007. From 2001 to 2002, he was a System Designer at STMicroelectronics. His current research interests include 5G/6G technologies, Internet of Things, communications security, machine learning, and intelligent communications. He has over 500 highly cited journals and conference papers, in addition to over 30 granted and pending patents and several standard contributions.

Dr. Wang is a Fellow of the Canadian Academy of Engineering and a Fellow of the Engineering Institute of Canada. He has received many prestigious awards and recognitions, including the IEEE Canada R. A. Fessenden Award, Canada Research Chair, Engineering Research Excellence Award at Western University, Canadian Federal Government Public Service Award, Ontario Early Researcher Award, and nine Best Paper Awards. He was involved in many IEEE conferences, including GLOBECOM, ICC, VTC, PIMRC, WCNC, CCECE, and CWIT, in different roles, such as General Chair, TPC Chair, Symposium Chair, Tutorial Instructor, Track Chair, Session Chair, and Keynote Speaker. He serves/has served as the Editor-in-Chief, Associate Editor-in-Chief, and editor/associate editor for over ten journals. He was the Chair of the IEEE ComSoc Signal Processing and Computing for Communications (SPCC) Technical Committee and is currently serving as the Central Area Chair for IEEE Canada.

\end{IEEEbiography}

\vfill

\end{document}